\documentclass[11pt]{article}
\pdfoutput=1 

\usepackage{jheppub} 

\usepackage[T1]{fontenc} 

\usepackage{amsmath}
\usepackage{amsthm}
\usepackage{amssymb}
\usepackage{graphicx}
\usepackage{wasysym}
\usepackage{amsfonts}
\usepackage{bm}
\usepackage{enumerate}
\usepackage{color}
\usepackage{hyperref}
\usepackage{simpler-wick} 
\usepackage{times}
\usepackage{indentfirst}
 
\newcommand{\beq}{\begin{equation}}
\newcommand{\eeq}{\end{equation}}
\newcommand{\bea}{\begin{eqnarray}}
\newcommand{\eea}{\end{eqnarray}}

\newcommand{\la}{\langle}
\newcommand{\ra}{\rangle}

\mathchardef\nss="711B

\def\nss{\mathcal{S}}

\def\be{\begin{eqnarray}}
\def\ee{\end{eqnarray}}

\newlength{\myL}

\theoremstyle{plain}
  \newtheorem{theorem}{Theorem}
  
  \newtheorem{lemma}[subsubsection]{Lemma}
  \newtheorem{corollary}{Corollary}
   
\theoremstyle{remark}
  \newtheorem{example}[subsubsection]{Example}

\title{\boldmath Decomposition of $\mathcal{N}=1$ superconformal minimal models and their fractional quantum Hall wavefunctions}
\author[a]{Yichen Hu,}
\author[b]{Sirui Ning,}
\author[c]{Yehao Zhou}

\affiliation[a]{Department of Physics, Princeton University,\\Princeton, New Jersey 08544, USA}
\affiliation[b]{The Rudolf Peierls Centre for Theoretical Physics, University of Oxford,\\Oxford OX1 3PU, United Kingdom}
\affiliation[c]{Kavli Institute for the Physics and Mathematics of the Universe (WPI), the University of Tokyo,\\Kashiwa, Chiba 277-8583, Japan}
\emailAdd{yh1553@princeton.edu}
\emailAdd{sirui.ning@physics.ox.ac.uk}
\emailAdd{yehao.zhou@ipmu.jp}

\abstract{$\mathcal{N}=1$ superconformal minimal models are the first series of unitary conformal field theories (CFTs) extending beyond Virasoro algebra. Using coset constructions, we characterize CFTs in $\mathcal{N}=1$ superconformal minimal models using combinations of a parafermion theory, an Ising theory and a free boson theory. Supercurrent operators in the original theory also becomes sums of operators from each constituent theory. If we take our $\mathcal{N}=1$ superconformal theories as the neutral part of the edge theory of a fractional quantum Hall state, we present a systematic way of calculating its ground state wavefunction using free field methods. Each ground state wavefunction is known previously as a sum of polynomials with distinct clustering behaviours. Based on our decomposition, we find explicit expressions for each summand polynomial. A brief generalization to $S_3$ minimal models using coset construction is also included. }

\begin{document}

\maketitle
\section{Introduction}
 Unitary conformal field theories \cite{Belavin:1984vu} have played important roles in many different fields of physics ranging from string theory \cite{Friedan1986,DiFrancesco:1997nk} to critical statistical mechanical systems \cite{Cardy,DiFrancesco:1997nk}. Recent studies have also shown their intimate connections to topological phases of matter from symmetry protected topological phases \cite{Shinsei1,Cho2017,Teo1,Cenke1} to fractional quantum Hall states \cite{Halperin2020,Moore:1991ks,Read_2009,Simon_2007,Simon_2009,Simon_2010, Read:1998ed,Bernevig_2008}. Specifically, for quantum Hall states, following many pioneering works, wavefunctions of fractional quantum Hall states are shown to be closely related to correlators of primariy fields of the unitary conformal field theory describing its low-energy edge physics. This is a manifestisation of the general ``bulk-boundary'' correspondence of a topologically ordered state where the bulk hosts a topological field theory and the boundary hosts a conformal field theory. Wavefunctions from unitary conformal field theories associated with fractional quantum Hall states provide deep insights into understandings of  physical properties of these exotic phases of matter, such as anyonic excitaions, fractional statistics and even non-Abelian fusion structures among anyons.

$\mathcal{N}=1$ superconformal field theories (SCFTs) are the first series of unitary field theories non-trivially extending Virasoro algebra \cite{Goddard1986,Qiu1986,Friedan1985}. Each member SCFT has a $\Delta=\frac{3}{2}$ supercurrent operator $G$ (also called $\mathbb{Z}^{(r)}_2$ parafermion in \cite{Estienne_2010} because of its fusion structure: $G\times G=\mathbb{I}$) that together with the usual $\Delta=2$ stress-energy tensor $T$  forming the $\mathcal{N}=1$ superconformal algebra. This algebra coincides with $WB_1(\beta)$ algebra \cite{Lukyanov:1990tf}, where $\beta$ is related to the central charge by $c=15/2-3\beta-3\beta^{-1}$. Note that $WB_1(\beta)$ minimal models correspond to $\beta=p/q$, where $p,q$ are coprime integers. Other generalizations beyond Virasoro algebra involving operators with fractional spin $\Delta=\frac{N+1}{N}$ can be found in \cite{Chung:1991kb,Argyres_1993,Argyres_1994} ($N=3$ corresponds to $S_3$ minimal models \cite{Fateev}). Our motivations to study this particular series of CFTs are two fold. Recent works following \cite{Harvey2018,Harvey2020,Bae2021,Duan2022} have revealed an intricate structure in the theory space of rational CFTs. Different rational CFTs with highly similar topological modular data, fusion structures and anyon contents, are related by number theoretical maps at the level of characters. One can trace these different classes of topological modular data back to realizations in either some parafermion or minimal model $M(p,q)$ theories (usually non-unitary $q\neq p+1$). Along this line of thinking, we want to understand how our $\mathcal{N}=1$ minimal models can be decomposed into more ``primitive'' CFTs. Furthermore, we go beyond characters and try to find explicit decompositions of primary fields with Abelian fusion structure. On the other hand, certain anyons are essential for topological quantum computation \cite{Freedman2002}. Member CFTs in $\mathcal{N}=1$ series are known to host these anyons that are capable of performing universal quantum computation by braiding alone. For example, the tri-critical Ising model \cite{DiFrancesco:1997nk}, also known as the edge theory of the ``anti-Fibonacci'' phase \cite{Kane1}, hosts Fibonacci anyons with scaling dimension $\Delta=\frac{3}{5}$ and fusion rule $\tau \times \tau=\mathbb{I}+\tau$. A comprehensive understanding of the ground state and quasiparticle wavefunctions of these topological phases are crucial for their quantum information applications.

This paper is organized as follows: in section ~\ref{2}, we present explicit characterization of the $\mathcal{N}=1$ superconformal minimal models using coset construction. This also gives us an unique decomposition of the supercurrent operators. In section~\ref{3}, we take our $\mathcal{N}=1$ superconformal field theories as the neutral part of a fractional quantum Hall state. We work out its ground state wavefunction based on correlators of the supercurrent operators. The main tool is to use the fact that correlators of supercurrent operators depends on the central charge in a polynomial way, so correlators of supercurrent operators in various free field theories with different central charges uniquely determine the correlators themselves. By fixing the limit theory of our minimal models at $c=\frac{3}{2}$,  we can further organize our wavefunction polynomial into sums of symmetric polynomials with distinct clustering behaviours following \cite{Estienne_2010}. Our approach gives explicit expression for each summand symmetric polynomial in a concise format.



\section{$\mathcal{N}=1$ superconformal minimal models and its coset constructions}
\label{2}
$\mathcal{N}=1$ superconformal minimal models are a series of conformal field theories with $\mathcal{N}=1$ super-Virasoro algebra generated by 
\begin{equation}
\begin{aligned}
T(z)T(z^\prime)=&\frac{c/2}{\left(z-z^\prime\right)^4}+\frac{2T(z^\prime)}{\left(z-z^\prime\right)^2}+\frac{\partial T(z^\prime)}{z-z^\prime}+\mathcal{O}(1)\\
T(z)G(z^\prime)=&\frac{(3/2)G(z^\prime)}{\left(z-z^\prime\right)^2}+\frac{\partial G(z^\prime)}{z-z^\prime}+\mathcal{O}(1)\\
G(z)G(z^\prime)=&\frac{(2c/3)}{\left(z-z^\prime\right)^3}+\frac{2T(z^\prime)}{z-z^\prime}+\mathcal{O}(1)
\end{aligned}
\end{equation} 
where $z=x+i\tau$ denoting the $(1+1)$d Euclidean space-time coordinates compactified on a infinite cylinder and $G$ is the scaling dimension $\Delta=\frac{3}{2}$ supercurrent operator that generates the $\mathcal{N}=1$ supersymmetry. They have central charge $c=\frac{3}{2}(1-\frac{8}{m(m+2)})$ with $m=3,4,\cdots$ where $m=3$ is the tri-critical Ising CFT. Each member SCFT has a known coset construction \cite{Goddard1986,DiFrancesco:1997nk}
\beq
\frac{SU(2)_2\times SU(2)_{m-2}}{SU(2)_m}
\eeq 
\subsection{Decomposition}
In order to decompose our $\mathcal{N}=1$ theories, let us first large our SCFTs by pairing them with a $\mathbb{Z}_m$ parafermion CFT $\frac{SU(2)_m}{U(1)_{2m}}$. Reshuffling factors in cosets of this product theory, we get the following relations:
\begin{equation}
\frac{SU(2)_m}{U(1)_{2m}}\times \frac{SU(2)_2\times SU(2)_{m-2}}{SU(2)_m}=\frac{SU(2)_{m-2}}{U(1)_{2(m-2)}}\times  \frac{SU(2)_2}{U(1)_4}\times U(1)_{4m(m-2)}
\label{N=1}
\end{equation} 
This equality holds at the level of stress-energy tensor. The $U(1)$ factors in Eq. (\ref{N=1}) needs a little explanation. Let us reshuffle $U(1)$ theories as following:
\beq
\text{Left}\quad (U(1)_{2m})^{-1}\times U(1)_{2(m-2)}\times U(1)_{4}, \quad\text{Right} \quad U(1)_{4m(m-2)}.
\label{U(1)}
\eeq To establish an equivalence relation between the left and right hand side, we first write out their Lagrangian densities:
\beq
\mathcal{L}_{l}=\frac{1}{4\pi} \sum_{I,J=1}^{3}K_{l}^{IJ} \partial_x \phi_{I} (\partial_x+\partial_t) \phi_{J}, \quad \mathcal{L}_{r}=\frac{1}{4\pi} K_{r} \partial_x \phi_1 (\partial_x+\partial_t) \phi_1
\eeq where $K_l=\begin{pmatrix}
    -2m & 0 &0\\
    0 & 2(m-2) & 0\\
    0 & 0 & 4
\end{pmatrix}$ and $K_r=4m(m-2)$. It is easy to see that there exists a matrix $M$ with integer entries
\beq
M=\begin{pmatrix}
2-m&-m&0\\
1&1&1\\
-1&-1&1\\
\end{pmatrix}, \quad \det{M}=4>0
\eeq  such that
\beq
M^T
(\begin{pmatrix}
-2m & 0 & 0\\
0 & 2(m-2) & 0\\
0 & 0 & 4\\
\end{pmatrix})
M=\begin{pmatrix}
4m(m-2)\\
\end{pmatrix}\bigoplus 8\sigma_z.
\label{K-matrix}
\eeq 
The right hand side of Eq.~(\ref{K-matrix}) is equivalent to $\mathcal{L}_r$ as we can add a backscattering term to the two counter-propagating modes
\beq
2\cos{(\sqrt{2}(\phi_2-\phi_3))}
\eeq to gap out $8\sigma_z$ degrees of freedom ($\phi_2$ and $\phi_3$)\cite{leonid,Kane1992}. We have established an equivalence relation in Eq. (\ref{U(1)}).



Eq. (\ref{N=1}) means we can combine a $\mathbb{Z}_m$ parafermion theory with the $m$th member of the $\mathcal{N}=1$ SCFT theory and turn them into a linear combination of a $\mathbb{Z}_{m-2}$ parafermion theory, a $U(1)$ free boson theory and an Ising ($\mathbb{Z}_2$ parafermion) theory. Focusing on stress-energy tensor, we have
\beq
\begin{aligned}
T_{min}&=\gamma T_{\mathbb{Z}_{m-2}}+\alpha T_{\mathbb{Z}_2}+\beta T_{U(1)}+ \delta t\\
T_{\mathbb{Z}_m}&=(1-\gamma) T_{\mathbb{Z}_{m-2}}+(1-\alpha) T_{\mathbb{Z}_2}+(1-\beta) T_{U(1)}- \delta t
\end{aligned}
\eeq where $\gamma$, $\alpha$, $\beta$ and $\delta$ are coefficients representing weights for each constituent stress-energy tensor. $t$ is an extra $\Delta=2$ auxiliary field made out of Abelian primary fields from each constituent CFT ($\gamma$ from Ising theory, $\varphi_1$ ($\bar{\varphi_1}$) from parafermion theory and $e^{i\sqrt{\frac{m}{m-2}}\phi}$ from $U(1)$ theory), that is key to make this decomposition work.
A useful way to understand the appearance of the auxiliary spin-$2$ field $t$ is the following.
In the product theory $\mathbb{Z}_{m-2}\times \mathbb{Z}_2 \times U(1)$, we would like to find two
mutually commuting stress tensors $T_{\min}$ and $T_{\mathbb{Z}_m}$ whose sum equals the total
stress tensor $T_{\mathbb{Z}_{m-2}}+T_{\mathbb{Z}_2}+T_{U(1)}$.
If one restricts to linear combinations of the three constituent stress tensors only, the OPE closure
conditions for $T_{\min}$ and $T_{\mathbb{Z}_m}$ are generically overconstrained and do not admit
a solution compatible with mutual commutativity.
The role of $t$ is precisely to provide an additional independent direction in the neutral,
local spin-$2$ sector so that the decomposition problem becomes solvable.

Our choice of $t$ is further constrained by locality and neutrality:
we take $t$ to be a local dimension-$2$ operator built from Abelian (simple-current) primaries of
each constituent CFT, so that it braids trivially with the fields that appear in the decomposition.
With this requirement, and demanding that the $t\times t$ OPE closes onto the stress tensors as in
Eq.~(2.10), the operator $t$ is fixed up to an overall normalization and sign (which can be absorbed
into $\delta$). Other dimension-$2$ operators are either Virasoro descendants/improvement terms that
can be reabsorbed into a redefinition of the constituent stress tensors, or they are non-local with
respect to the desired operator algebra and are therefore excluded from the stress-tensor decomposition.
$t$ satisfies
\beq
t(z)t(w)\sim\frac{1}{\left(z-w\right)^4}+\frac{1}{\left(z-w\right)^2}
\left(2T_{\mathbb{Z}_2}+\frac{m}{m-2}T_{U(1)}+\frac{m}{m-2}T_{\mathbb{Z}_{m-2}}\right)
\eeq since consistency of operator product expansion between stress-energy tensor $T$'s and $t$ requires coefficient in front of $T_{\mathbb{Z}_2}$ is $\frac{2\Delta(\gamma)}{c_{\mathbb{Z}_2}}=2$, $T_{U(1)}$ is $\frac{2\Delta(e^{i\sqrt{\frac{m}{m-2}}\phi})}{c_{U(1)}}=\frac{m}{m-2}$ and $T_{\mathbb{Z}_{m-2}}$ to be $\frac{2\Delta(\varphi_1 (\bar{\varphi_1}))}{c_{\mathbb{Z}_{m-2}}}=\frac{m}{m-2}$.

Using this fact together with the usual stress-energy tensor OPE: 
\begin{equation}
   T_{\textrm{min}/\mathbb{Z}_m}(z)T_{\textrm{min}/\mathbb{Z}_m}(w)\sim \frac{c/2}{\left(z-w\right)^4}+\frac{2T_{\textrm{min}/\mathbb{Z}_m }(w)}{\left(z-w\right)^2}, 
\end{equation}
we arrive at the following set of equations for these coefficients:
\beq
\begin{aligned}
&\frac{1}{2}\alpha+\frac{m}{2(m-2)}\beta+\frac{m-3}{m-2}\gamma=1 \textrm{ (for $t$ coefficient)}\\
&\alpha^2+\delta^2=\alpha\textrm{ (for $\mathbb{Z}_2$ coefficient)}\\
&\beta^2+\frac{m}{2(m-2)}\delta^2=\beta\textrm{ (for $U(1)$ coefficient)}\\
&\gamma^2+\frac{m}{2(m-2)}\delta^2=\gamma\textrm{ (for $\mathbb{Z}_{m-2}$ coefficient)}
\end{aligned}
\eeq They have two solutions which corresponds to the two stress-energy tensors $T_{min}$ and $T_{\mathbb{Z}_m}$:
\beq
\begin{split}
&T_{min}= \frac{m-2}{m+2}T_{\mathbb{Z}_2}+\frac{m}{m+2}T_{U(1)}+\frac{2}{m+2}T_{\mathbb{Z}_{m-2}}\pm\frac{2\sqrt{m-2}}{m+2}t \\
&T_{\mathbb{Z}_m}=\frac{4}{m+2}T_{\mathbb{Z}_2}+\frac{2}{m+2}T_{U(1)}+\frac{m}{m+2}T_{\mathbb{Z}_{m-2}}\mp\frac{2\sqrt{m-2}}{m+2}t
\end{split}
\eeq 
Equation (2.12) admits a discrete ``exchange'' symmetry that maps one solution of (2.13) to the other.
Concretely, the transformation
\[
(\alpha,\beta,\gamma,\delta)\ \longrightarrow\ (1-\alpha,\ 1-\beta,\ 1-\gamma,\ -\delta)
\]
leaves Eq.~(2.9) invariant but exchanges the two commuting Virasoro generators,
$T_{\min}\leftrightarrow T_{\mathbb{Z}_m}$, while keeping their sum fixed.
Thus the two branches in Eq.~(2.13) should be viewed as the two possible assignments of the two
mutually commuting stress tensors inside the enlarged theory.

Our naming convention is fixed by the $\mathcal{N}=1$ superconformal structure:
by definition, $T_{\min}$ is the stress tensor that appears in the $\mathcal{N}=1$ super-Virasoro algebra
together with the supercurrent $G$, whereas $T_{\mathbb{Z}_m}$ is the commuting spectator sector.
Equivalently, among the two branches in Eq.~(2.13), we choose the one for which
$T_{\min}(z)G(w)$ has the standard primary OPE with conformal weight $3/2$ and
$T_{\mathbb{Z}_m}(z)G(w)$ is regular, as imposed in Eq.~(2.14).
This criterion removes the ambiguity and fixes the assignment unambiguously (up to an overall sign
convention for $t$).

Now, we further decompose the supercurrent operator $G$ in terms of operators from different sectors. To this end, we have
\beq
\begin{aligned}
&T_{min}(z)G(w)\sim\frac{\frac{3}{2}G(w)}{\left(z-w\right)^2}\quad T_{\mathbb{Z}_m}(z)G(w)\sim 0\\
&G(z)G(w)\sim\frac{\frac{2c}{3}}{\left(z-w\right)^3}
\end{aligned}
\eeq
These relations give us a unique decomposition of $G$ with the exception of $m=4$ (it has $\mathcal{N}=2$ superconformal algebra and therefore the decomposition of $G$ is not unique due to the emergent $U(1)$ $R$-symmetry). To derive Eqs.~(2.15) and (2.16) systematically, we start from an operator-basis viewpoint.
In the product theory $\mathbb{Z}_{m-2}\times \mathbb{Z}_2 \times U(1)$, locality and neutrality
severely restrict the possible dimension-$3/2$ operators that can contribute to the supercurrent.
Up to overall normalization, the relevant local candidates are
\[
\varphi_1\,e^{+i a\phi},\qquad \bar{\varphi}_1\,e^{-i a\phi},\qquad \gamma\,\partial\phi,
\qquad a=\sqrt{\frac{m}{m-2}},
\]
which all have scaling dimension $3/2$.
We therefore write the most general ansatz as a linear combination of these operators.
The discrete ``particle-hole'' symmetry in Eq.~(2.17) rules out the symmetric combination
$\varphi_1 e^{i a\phi}+\bar{\varphi}_1 e^{-i a\phi}$, leaving an ansatz of the form
\[
G \;=\; A\Big(\varphi_1 e^{i a\phi}-\bar{\varphi}_1 e^{-i a\phi}\Big)\;+\;B\,\gamma\,\partial\phi.
\]
Imposing the defining OPE constraints in Eq.~(2.14) (namely, that $G$ is a weight-$3/2$ primary under
$T_{\min}$, is neutral under $T_{\mathbb{Z}_m}$, and that $G\times G$ reproduces the correct
$\mathcal{N}=1$ super-Virasoro OPE) fixes the coefficients uniquely up to an overall sign.
The difference between even and odd $m$ arises from locality/monodromy conditions, which determine
the relative phases needed for $G$ to be a single-valued local field.
For $m$ odd, 
\begin{equation}
\begin{aligned}
&G=\pm i\sqrt{\frac{2(m-2)}{m(m+2)}}\left(\varphi_1e^{i\sqrt{\frac{m}{m-2}}\phi}-
\bar\varphi_1 e^{-i\sqrt{\frac{m}{m-2}}\phi}\right)\mp i\sqrt{\frac{m-2}{m+2}}\gamma \partial \phi\\
&t=\frac{i}{\sqrt{2}}\left(\varphi_1\gamma e^{i\sqrt{\frac{m}{m-2}}\phi}+\bar\varphi_1\gamma e^{-i\sqrt{\frac{m}{m-2}}\phi}\right)
\end{aligned}
\label{G}
\end{equation} and for $m$ even, 
\beq
\begin{aligned}
&G=\pm\sqrt{\frac{2(m-2)(m+4)}{m(m+2)(m-4)}}\left(\varphi_1 e^{i\sqrt{\frac{m}{m-2}}\phi}-\bar\varphi_1 e^{-i\sqrt{\frac{m}{m-2}}\phi}\right)\pm i \sqrt{\frac{(m-2)(m+4)}{(m+2)(m-4)}}\gamma\partial \phi\\
&t=\frac{1}{\sqrt{2}}\left(\varphi_1\gamma e^{i\sqrt\frac{m}{m-2}\phi}+\bar\varphi_1\gamma e^{-i\sqrt{\frac{m}{m-2}}\phi}\right)
\end{aligned}
\eeq where $\{\varphi_0=\mathbb{I},\varphi_1,\varphi_2,\cdots,\varphi_{m-3}=\bar{\varphi}_1\}$ are the $\mathbb{Z}_{m-2}$ parafermion primary fields, $\gamma$ is a Majorana fermion from Ising CFT and $\phi$ is a $U(1)$ boson field. Parafermions (including $\mathbb{Z}_{m-2=2}$,  Majorana fermion) $\varphi_i$ has scaling dimension $\Delta=\frac{i(m-2-i)}{m-2}$ and fusion rule $\varphi_i \times \varphi_j=\varphi_{(i+j) \text{ mod } (m-2)}$. Boson field has compactification radius $\phi=\phi+2\pi$. Its vertex operators $e^{ia \phi}$ has  scaling dimension $\Delta=\frac{a^2}{2}$. From these relations, we can conclude that $\Delta(\varphi_1 e^{i\sqrt{\frac{m}{m-2}}\phi})=\Delta(\bar{\varphi}_1 e^{-i\sqrt{\frac{m}{m-2}}\phi})=\Delta(\gamma \partial \phi)=\frac{3}{2}$.

Curiously, $\varphi_1 e^{i\sqrt{\frac{m}{m-2}}\phi}$ ($\bar\varphi_1 e^{-i\sqrt{\frac{m}{m-2}}\phi}$) appeared in Eq. (\ref{G}) is the same as the electron operator of the $(m-2)$th member in Read-Rezayi quantum Hall states with filling factor $\nu=\frac{m-2}{m}$ even though $\varphi_1 e^{i\sqrt{\frac{m}{m-2}}\phi}$ ($\bar\varphi_1 e^{-i\sqrt{\frac{m}{m-2}}\phi}$) does not carry any physical charge. In low energy, it is known that Read-Rezayi states at filling factor $\nu=\frac{m-2}{m}$ possess emergent $\mathcal{N}=2$ supersymmetry with the physical $U(1)$ charge being its $R$-symmetry \cite{Santos}. It is therefore not surprising to see the appearance of $\varphi_1 e^{i\sqrt{\frac{m}{m-2}}\phi}-\bar\varphi_1 e^{-i\sqrt{\frac{m}{m-2}}\phi}$ in decomposition of $\mathcal{N}=1$ supercurrent operator $G$.  The other generator for the other half of supersymmetry in this rotated basis is $\varphi_1 e^{i\sqrt{\frac{m}{m-2}}\phi}+\bar\varphi_1 e^{-i\sqrt{\frac{m}{m-2}}\phi}$. We observe that the following ``particle-hole'' (since there is no real physical charge) symmetry 
\beq
\phi\Longleftrightarrow -\phi,~ \varphi_1 \Longleftrightarrow  -\bar{\varphi_1}, ~ \gamma \Longleftrightarrow  -\gamma
\eeq takes $G\rightarrow G$ and forbids $\varphi_1 e^{i\sqrt{\frac{m}{m-2}}\phi}+\bar\varphi_1 e^{-i\sqrt{\frac{m}{m-2}}\phi}$.

Lastly, from our decomposition, we see that as $m\rightarrow \infty$, central charge of SCFTs goes as $c\rightarrow \frac{3}{2}$. Operator-wise, $T_{min}(z)\rightarrow T_{\mathbb{Z}_2}+T_{U(1)}$ and $G\rightarrow \gamma\partial\phi$ which makes it clear that the conformal field theory in the infinity limit is the three Majorana fermion theory which is on the moduli space of $c=\frac{3}{2}$ circle line \cite{Dixon1988}. This is analogous to the unitary minimal models where they become the $U(1)$ free boson theory on the moduli space of $c=1$ circle line in the $m\rightarrow \infty$ limit \cite{Ginsparg}. 

\section{Fractional Quantum Hall wavefunction}
\label{3}
To study fractional quantum Hall states based on $\mathcal{N}=1$ minimal models, we need to combine our $\mathcal{N}=1$ minimal models (the charge neutral sector) with a $U(1)$ physical charge sector. First, let us define an electron operator that has trivial fusion structure and also braid trivially with all other quasi-particles. Formally, we can write this electron operator as $\psi_e=Ge^{i\sqrt{p/q}\Phi}$ where $G$ is the supercurrent operator of our $\mathcal{N}=1$ theories and $\Phi$ is the boson field representing the $U(1)$ charge. $p, q$ coprime $\in \mathbb{Z}^+$ gives a filling factor of $\frac{q}{p}$ for our FQH state (assuming our ground state only occupy the lowest Landau level). By the bulk-boundary correspondence of quantum Hall physics, its ground state wavefunction is the same as the chiral correlation function from the edge $\mathcal{N}=1$ SCFT and the $U(1)$ theory:
\beq
\la\psi_e(z_1)\cdots \psi_e(z_{N}) \ra=\la G(z_1)\cdots G(z_N)\ra \la e^{i\sqrt{p/q}\Phi(z_1)}\cdots e^{i\sqrt{p/q}\Phi(z_N)} \ra=C_N \prod_{i<j}^N(z_i-z_j)^{p/q} 
\eeq where $p/q$ is chosen such that $\la\psi_e(z_1)\cdots \psi_e(z_{N}) \ra$ is a polynomial with no pole and branch-cut. From this, we see that the task of computing this wavefunction is really about computing $C_N\equiv \la G(z_1)\cdots G(z_N)\ra$.
\subsection{Correlators of supercurrent operator G}
In this section, we compute the $n$ point correlators of supercurrent operators $G$ using free field methods \cite{10.1215/S0012-7094-92-06604-X,Gaberdiel:1998fs,Gaberdiel:1999mc}. Physically, free field methods utilize free bosonic, fermionic or ghost theories to represent a two dimensional conformal field theory. The simplest example is the well-known "Coulomb-Gas" formalism. Free field methods allow us to compute correlation functions of conformal fields via Wick's theorem.\smallskip

Before calculations for the supercurrent $G$, as a demonstration of principle, we apply free field methods on stress-energy tensor $T$ to calculate its $n$ correlation functions. 
We start from Virasoro algebra, which is generated by:
\begin{equation}
    T(z)T(z^\prime)=\frac{c/2}{\left(z-z^\prime\right)^4}+\frac{2T(z^\prime)}{\left(z-z^\prime\right)^2}+\frac{\partial T(z^\prime)}{z-z^\prime}+\mathcal{O}(1)\label{OPE1}
\end{equation}
For the $n$ point correlator of stress energy tensors, we have the following theorem:
\begin{theorem}
\begin{equation}
\la T(z_{1})T(z_{2})\cdots T(z_{n})\ra=\mathop{\mathop{\sum} \limits_{\sigma\in S_{n}}}\limits_{\sigma=(l_{1})\cdots(l_{s})}\mathop{\prod}\limits_{i=1}^{s}f_{l_{i}}\label{them1},
\end{equation}
where the perumutation $\sigma=(l_{1})...(l_{s})\in S_{n}$ is  equivalent to product of  cyclic permutations of length at least 2, and for each cyclic permutation $(l_{i})=(i_{1}i_{2}...i_{m_{l}})$, define 
\begin{equation}
f_{l_{i}}(z_{i_{1}},z_{i_{2}},\cdots,z_{i_{m_{l}}})=\frac{c/2}{(z_{i_{1}}-z_{i_{2}})^{2}(z_{i_{2}}-z_{i_{3}})^{2}\cdots(z_{i_{m_{l}}}-z_{i_{1}})^{2}}.
\end{equation}
\end{theorem}
We start the proof with two lemmas:

\begin{lemma}\label{lemma: T correlators are rational functions}
$\la T(z_{1})T(z_{2})\cdots T(z_{n})\ra$ is a rational function in variables $z_{1},z_{2},\cdots,z_{n}$ and central charge $c$.
\end{lemma}

\begin{proof}
This lemma can be proven by mathematical induction. For $n=2$, the two point correlators of stress-energy tensor is $\la T(z_{1})T(z_{2})\ra=\frac{c/2}{(z_{1}-z_{2})^{4}}$, which by definition is a rational function in variables $z_{1},z_{2}$ and $c$. \\
Suppose the conclusion holds for any $n-1$ point correlators $T_{n-1}(z_{2},z_{3},...z_{n})=\la T(z_{2})T(z_{3})...T(z_{n})\ra$. By Eq (\ref{OPE1}), the $n$ point correlators satisfy the following recursion relation:
\begin{equation}
    T_{n}(z_{1},...z_{n})=\mathop{\sum}\limits_{i=2}^{n}\left(\frac{c/2}{(z_{1}-z_{i})^{4}}T_{n-2}(\hat{z_{1}},\hat{z_{i}})+(\frac{2}{(z_{1}-z_{i})^{2}}+\frac{\partial_{i}}{z_{1}-z_{i}})T_{n-1}(z_{2},z_{3},...z_{n})\right).
\end{equation}
Where  $T_{n}(z_{1},...z_{n})=\la T(z_{1})T(z_{2})...T(z_{n})\ra$ and  $T_{n-2}(\hat{z_{1}},\hat{z_{i}})=\la T(z_{2})...T(z_{i-1}),T(z_{i+1})...T(z_{n})\ra$.
By mathematical induction, $T_{n-2}(\hat{z_{1}},\hat{z_{i}})$ and $T_{n-1}(z_{2},z_{3},\cdots ,z_{n})$ are both rational function, $T_{n}(z_{1},\cdots,z_{n})$ is also a rational function. Therefore the conclusion holds for $n$ point correlators. By mathematical induction, this conclusion holds for all integers $n\geq 2$.
\end{proof}

\begin{lemma}\label{lemma: identical rational functions}
Let $f(x_{1},x_{2},\cdots,x_{M})$ and $g(x_{1},x_{2},\cdots,x_{M})$ be two rational functions over $\{x_i\}_{i=1}^M\in \mathbb{C}$. If $$f(y,x_{2},\cdots,x_{M})=g(y,x_{2},\cdots,x_{M})$$ holds for infinitely many $y\in\mathbb{C}$, then $f=g$ identically.
\end{lemma}

\begin{proof}
Fixing $x_2,\cdots,x_M$, the function $f(x_1,x_2,\cdots,x_M)-g(x_1,x_2,\cdots,x_M)$ is a rational function in $x_1$, and it has infinitely many zeroes, but the number of zeroes of a rational function is bounded by the degree of its numerator, thus $f(x_1,x_2,\cdots,x_M)-g(x_1,x_2,\cdots,x_M)=0$ identically. 
\end{proof}

Let 
\begin{equation}
F(c,z_{1},z_{2},\cdots,z_{n})=\mathop{\mathop{\sum} \limits_{\sigma\in S_{n}}}\limits_{\sigma=(l_{1})\cdots(l_{s})}\mathop{\prod}\limits_{i=1}^{n}f_{l_{i}}.
\end{equation}
If we can show that $$\la T(z_{1})T(z_{2})\cdots T(z_{n})\ra=F(c,z_{1},z_{2},\cdots,z_{n})$$ for infinitely many central charges $c$, then the theorem is proven.

At this point we emphasize that the ``free-field check'' is used only as a device to generate
infinitely many values of the central charge.
By Lemma~3.1.1, for fixed insertion points $\{z_i\}$ the correlator
$\langle T(z_1)\cdots T(z_n)\rangle$ is a rational function of $c$ (and of the $z_i$'s).
The closed-form expression $F(c,z_1,\dots,z_n)$ in Eq.~(3.6) is manifestly rational in $c$ as well.
Therefore, their difference is a rational function of $c$; if it vanishes for infinitely many values of $c$,
it must vanish identically by Lemma~3.1.2.
This justifies why proving the identity for the infinite sequence of free-field realizations
($c=N$ for free bosons, and similarly $c=3N$ in the $\beta\gamma$--$bc$ construction for supercurrents)
establishes the result for general $c$, without invoking any additional analytic-continuation assumptions.

\begin{proof}
We construct the free field realization of Virasoro algebra via free bosons $\phi_{i}$, $i=1,\cdots,N$ with OPE: 
\begin{equation}
\partial\phi_{i}(z)\partial\phi_{j}(w)=\frac{\delta_{ij}}{(z-w)^{2}}+O(1).
\end{equation}
The action for our free field realization is:
\begin{equation}
\int d^{2}z\mathop{\sum}\limits_{i=1}^{N}\frac{1}{2}\partial\phi_{i}\mathop{\partial}\limits^{-}\phi_{i}
\end{equation}
with stress energy tensor $T(z)=\mathop{\sum}\limits_{i=1}^{N}:\frac{1}{2}\partial\phi_{i}(z)\partial\phi_{i}(z):$. One can check that the OPE of this stress energy tensor satisfy Eq (\ref{OPE1}) with central charge $c=N$. Then the correlator $\la T(z_{1})T(z_{2})\cdots T(z_{n})\ra$ can be computed using Wick's theorem, namely we sum over all possible contractions between pair of free bosonic fields:
\begin{equation}
\frac{1}{2^{n}}\wick{\la :\c4\partial\phi_{i_{1}}(z_{1})\c1\partial\phi_{i_{1}}(z_{1}): :\c1\partial\phi_{i_{2}}(z_{2})\c2\partial\phi_{i_{2}}(z_{2}):.\c4...\c3..:\c2\partial\phi_{i_{k}}(z_{k})\c3\partial\phi_{i_{k}}(z_{k}):\ra}.
\end{equation}

Note that graphically, for all possible contractions, each vertex $z_{i}$ ($i=1,2,...,n$) must connect with other two vertices $z_{j}$ and $z_{k}$ ($j, k \neq i$). Therefore, for each possible contraction, we first decompose the $n$ vertices into different clusters, then we connect all the vertices in each cluster to form a loop. Physically, the amplitude associated with each contraction is the product of the amplitude of each loop, which is:
\begin{equation}
\frac{N/2}{(z_{i_{1}}-z_{i_{2}})^{2}\cdots(z_{i_{m_{l}}}-z_{i_{1}})^{2}}\times\frac{N/2}{(z_{j_{1}}-z_{j_{2}})^{2}\cdots(z_{j_{m_{j}}}-z_{j_{1}})^{2}}\times\cdots\times\frac{N/2}{(z_{s_{1}}-z_{s_{2}})^{2}\cdots(z_{s_{m_{s}}}-z_{s_{1}})^{2}}.
\end{equation}
The overall $n$-point amplitude is the sum of amplitudes of all possible contractions. Figure \ref{Figure 1} gives a graph representation of each possible contraction.\\
\begin{figure}
    \centering
    \includegraphics[scale=0.6]{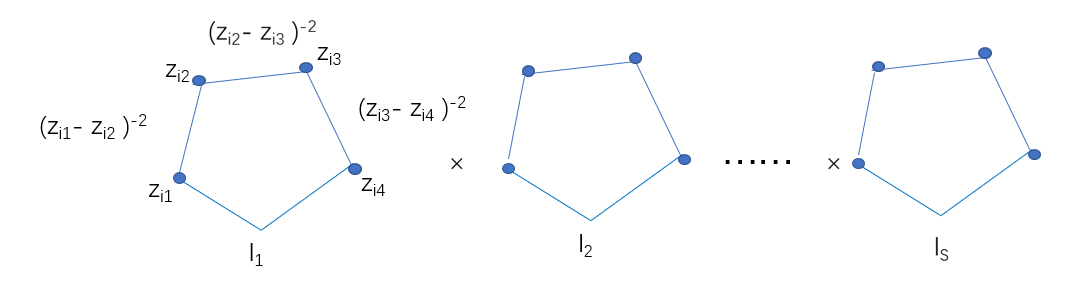}
    \caption{Every line represents a contraction between different pair of fields, $l_{i}$ labels a loop. By Wick's theorem, summing over all the possible products of loops gives the $n$ point amplitude.}
    \label{Figure 1}
\end{figure}

We give an interpretation to the factor $N/2$ associated with amplitude for each loop. It is easy to understand the factor of $N$ because there are $N$ free bosonic fields in the representation and the OPE between different fields vanishes. For a loop with $m$ vertices, each energy momentum tensor carries a factor of $1/2$ , which contributes a factor of $\frac{1}{2^{m}}$, and each double contraction gives a factor of 2, however, because the diagram is a loop, so all the double contractions contribute a factor of $2^{m-1}$. Multiply them together leads to a factor of $1/2$. Therefore the total factor is $N/2$.
Therefore $\la T(z_{1})\cdots T(z_{n})\ra_{c=N}=F(c=N,z_1,\cdots,z_n)$ for infinitely many $N\in\mathbb{Z}_{>0}$, thus the theorem holds by the Lemma \ref{lemma: identical rational functions}.
\end{proof}

Following Eq (\ref{them1}), an equivalent way of writing the correlator is:
\begin{equation}
    \la T(z_{1})\cdots T(z_{n})\ra=\mathop{\mathop{\sum} \limits_{I_{1}\cup I_{2}\cdots\cup I_{s}=(12\cdots n)}}\limits_{\text{cluster\,decomposition}}\limits \left(\frac{c}{2}\right)^{s}\mathop{\prod}\limits_{i=1}^{n}\hat{f}_{I_{i}},
\end{equation}
where \begin{equation}
\begin{aligned}
    &\hat{f_{I_{i}}}(z_{i_{1}},\cdots,z_{i_{M_{i}}})=\frac{1}{M_{i}}\mathop{\sum}\limits_{\sigma\in S_{M_{i}}}\frac{1}{(z_{\sigma(i_{1})}-z_{\sigma(i_{2})})^{2}(z_{\sigma(i_{2})}-z_{\sigma(i_{3})})^{2}\cdots(z_{\sigma(i_{M_{i}})}-z_{\sigma(i_{1})})^{2}}.
\end{aligned}
    \end{equation}
 $(z_{i_{1}},z_{i_{2}},...,z_{i_{M_{i}}})\in I_{i}$. We divide the $n$ vertices into $s$ clusters labeled as $I_{i}$($i=1,2,...s$), where each $I_{i}$ contains $|I_{i}|=M_{i}$ elements. Each cluster $I_{i}$ contains at least two vertices. Below we give two examples:\\
\begin{equation}
    \begin{aligned}
        &n=3, s=1, I_{1}=(123)\\
        &n=4, s=2, I_{1}=(12),I_{2}=(34),\quad I_{1}=(13),I_{2}=(24),\quad I_{1}=(14),I_{2}=(23)\\
    \end{aligned}
\end{equation}
The factor $\frac{1}{M_{i}}$ is because that the Feynman diagram corresponding to each cluster $I_{i}$ is invariant under $M_{i}$ elements in permutation group $S_{M_{i}}$, where each element $\sigma_{j=1,2,\dots,M_i}$ is:
\begin{equation}
\begin{aligned}
    &\sigma_{j}(z_{k})=z_{M_{i}-j+k},\quad \text{for}\quad k=1,2,...j\\
    &\sigma_{j}(z_{k})=z_{k-j}, \quad\text{for} \quad k=j+1,...,M_{i}.\\
\end{aligned}
\end{equation}

Now we use free field techniques to calculate the $n$ point correlation functions of supercurrent $G$. We start from $\mathcal{N}=1$ super Virasoro algebra, which is generated by:
\begin{equation}
    \begin{aligned}
T(z)T(w)=&\frac{c/2}{\left(z-w\right)^4}+\frac{2T(w)}{\left(z-w\right)^2}+\frac{\partial T(w)}{z-w}+\mathcal{O}(1)\\
T(z)G(w)=&\frac{(3/2)G(w)}{\left(z-w\right)^2}+\frac{\partial G(w)}{z-w}+\mathcal{O}(1)\\
G(z)G(w)=&\frac{(2c/3)}{\left(z-w\right)^3}+\frac{2T(w)}{z-w}+\mathcal{O}(1)\label{SUSY}
    \end{aligned}
\end{equation}
Our goal is to compute 
\begin{equation}
\la G(z_{1})G(z_{2})\cdots G(z_{n})\ra.
\end{equation}
Similar to the stress-energy correlators, we have the following:

\begin{lemma}\label{lemma: G correlators are rational functions}
$\la G(z_{1})G(z_{2})\cdots G(z_{n})\ra$ is a rational function in variables $z_{1},z_{2},\cdots,z_{n}$ and central charge $c$.
\end{lemma}

\begin{proof}
This lemma can be proven by mathematical induction. For $n=1$, the one point correlator of supercurrent vanishes, i.e $\la G(z_{1})=0\ra$. For $n=2$, the two point correlators of supercurrent is $\la G(z_{1})G(z_{2})\ra=\frac{2c/3}{(z_{1}-z_{2})^{3}}$, which is a rational function in variables $z_{1},z_{2}$ and $c$.\\
Suppose the conclusion holds for the $n-2$ point correlators $G_{n-2}(z_{3},z_{4},...z_{n})=\la G(z_{3})G(z_{4})...G(z_{n})\ra$, by Eq (\ref{SUSY}), the n point correlators satisfy the following recursion relation:
\begin{equation}
    G_{n}(z_{1},...z_{n})=\mathop{\sum}\limits_{i=2}^{n}\left((-1)^{i}\frac{2c/3}{(z_{1}-z_{i})^{3}}+(-1)^{i}\frac{1}{z_{1}-z_{i}}\mathop{\sum}\limits_{j\neq i,1}^{n}\left(\frac{3}{(z_{i}-z_{j})^{2}}+\frac{\partial_{j}}{z_{i}-z_{j}}\right)\right)G_{n-2}(\hat{z_{1}},\hat{z_{i}}),
\end{equation}
where $G_{n}(z_{1},...z_{n})=\la G(z_{1})G(z_{2})...G(z_{n})\ra$ and $G_{n-2}(\hat{z_{1}},\hat{z_{i}})=\la G(z_{2})...G(z_{i-1})G(z_{i+1})...G(z_{n})\ra$.
By mathematical induction, $G_{n-2}(\hat{z_{1}},\hat{z_{i}})$ is a  rational function, $G_{n}(z_{1},...z_{n})$ is also a rational function. Therefore the conclusion holds for $n$ point correlators. By mathematical induction, this conclusion holds for all integers $n\geq 3$. Therefore the lemma is proved.\\
Note that as a by-product, since $\la G_{1}(z_{1})\ra$ vanishes, by this recursion relation, $G_{n}(z_{1},...z_{n})=0$ if $n$ is odd. 
\end{proof}

To begin with, we construct free field representation of $\mathcal{N}=1$ super Virasoro algebra via free ghost fields \cite{Polchinski:1998rr}: $\beta\gamma-bc$ systems which combines two anticommuting fields $bc$  with two commuting fields $\beta\gamma$. Their weights are:
\begin{equation}
    \begin{aligned}
        &h_{b_{i}}=\lambda,\quad h_{c^{i}}=1-\lambda,\quad i=1,2,\cdots,N\\
        &h_{\beta_{j}}=\lambda-\frac{1}{2},\quad h_{\gamma^{j}}=\frac{3}{2}-\lambda,\quad j=1,2,\cdots,N\\
    \end{aligned}
\end{equation}
The action for our free field realization is:
\begin{equation}
    S=\frac{1}{2\pi}\int d^{2}z\mathop{\sum}\limits_{i=1}^{N}(b_{i}\overline{\partial} c^{i}+\beta_{i}\overline{\partial} \gamma^{i}),
\end{equation}
Their OPEs are:
\begin{equation}
    \begin{aligned}
        &b_{i}(z)c^{j}(w)=\frac{\delta_{i}^{j}}{z-w}+\mathcal{O}(1)\\
&\gamma^{i}(z)\beta_{j}(w)=\frac{\delta_{i}^{j}}{z-w}+\mathcal{O}(1)\\
    \end{aligned}
\end{equation}
The stress-energy tensor and supercurrent read:
\begin{equation}
    \begin{aligned}
        &T=(\partial b_{i}(z))c^{i}(z)-\lambda\partial(b_{i}(z)c^{i}(z))+(\partial\beta_{i}(z))\gamma^{i}(z)-\frac{1}{2}(2\lambda-1)\partial(\beta_{i}(z)\gamma^{i}(z))\\
        &G=-\frac{1}{2}(\partial\beta_{i})(z)c^{i}(z)+\frac{2\lambda-1}{2}\partial(\beta_{i}(z)c^{i}(z))-2b_{i}(z)\gamma^{i}(z)\\
    \end{aligned}
\end{equation}
One can check that their OPE satisfy Eq (\ref{SUSY}). Without loss of generality we choose $\lambda=\frac{1}{2}$, then the supercurrent reads:
\begin{equation}
    G(z)=-\frac{1}{2}:\partial\beta_{i}(z)c^{i}(z):-2:b_{i}(z)\gamma^{i}(z):\label{super}
\end{equation}
with central charge $c=3N$.\\
\begin{theorem}
    \begin{equation}
        \la G(z_{1})G(z_{2})...G(z_{2n})\ra=\mathop{\mathop{\sum} \limits_{\sigma\in S_{2n}}}\limits_{\sigma=(l_{1})\cdots(l_{s})}\mathop{\prod}\limits_{i=1}^{s}\mathrm{sign}\begin{pmatrix}
    1&2&\cdots&2n\\
    l_{1}&l_{2}&\cdots&l_{s}\\
\end{pmatrix}g_{l_{i}}\label{them3},
\end{equation}
where the perumutation $\sigma=(l_{1})...(l_{s})\in S_{2n}$ is  equivalent to product of  cyclic permutations of length at least 2, and for each cyclic permutation $(l_{i})=(i_{1}i_{2}...i_{2m_{l}})$, define 
\begin{equation}
\begin{split}
&g_{l_i}(z_{i_{1}}\cdots z_{i_{2m}})=\frac{c/3}{(z_{i_{1}}-z_{i_{2}})(z_{i_{2}}-z_{i_{3}})^{2}\cdots(z_{i_{2m-1}}-z_{i_{2m}})(z_{i_{2m}}-z_{i_{1}})^{2}}\\
&-\frac{c/3}{(z_{i_{1}}-z_{i_{2}})^{2}(z_{i_{2}}-z_{i_{3}})\cdots(z_{i_{2m-1}}-z_{i_{2m}})^{2}(z_{i_{2m}}-z_{i_{1}})}.
\end{split}
\end{equation}
\end{theorem}

Let
\begin{equation}
H(c,z_{1},z_{2},...z_{2n})=\mathop{\mathop{\sum} \limits_{\sigma\in S_{2n}}}\limits_{\sigma=(l_{1})\cdots(l_{s})}\mathop{\prod}\limits_{i=1}^{s}\mathrm{sign}\begin{pmatrix}
    1&2&\cdots&2n\\
    l_{1}&l_{2}&\cdots&l_{s}\\
\end{pmatrix}g_{l_{i}}.
\end{equation}
If we can show that:
\begin{equation}
    \la G(z_{1})G(z_{2})...G(z_{2n})\ra=H(c,z_{1},z_{2},...z_{2n})
\end{equation}
for infinitely many central charges $c$, then the theorem is proven.\\
\begin{proof}
We know by free field realization $\la G(z_{1})\cdots G(z_{n})\ra$ can be computed by Wick's theorem. Note that there are two differneces now: first, the $n$ point correlator of $G$ vanishes for odd $n$, so we only need to consider the case when $n=2M$. Each loop must contain even number of vertices, otherwise the amplitude associated with the loop will vanish in the OPE. Second, for each contraction between pair of supercurrents, as the supercurrent Eq (\ref{super}) has two terms, there are two propagators between each pair of vertices. We write them down explicity:
\begin{equation}
    \begin{aligned}
    &\wick{\la :\c4\partial\beta_{i_{1}}(z_{1})\c1 c_{i_{1}}(z_{1}): :\c1 b_{i_{2}}(z_{2})\c2\gamma_{i_{2}}(z_{2}):.\c4..\c5.\c6.\c3.:\c2\partial\beta_{i_{2M-1}}(z_{2M-1})\c3 c_{i_{2M-1}}(z_{2M-1})\c5 b_{i_{2M}}(z_{2M})\c6\gamma_{i_{2M}}(z_{2M}):\ra},\\
    &\wick{\la :\c4 b_{i_{1}}(z_{1})\c5\gamma_{i_{1}}(z_{1}): :\c2 \partial\beta_{i_{2}}(z_{2})\c1 c_{i_{2}}(z_{2}):.\c7.\c1.\c5.\c6.\c3.:\c7 b_{i_{2M-1}}(z_{2M-1})\c2 \gamma_{i_{2M-1}}(z_{2M-1})\c6 \partial\beta_{i_{2M}}(z_{2M})\c4 c_{i_{2M}}(z_{2M}):\ra}.\\     
    \end{aligned}
\end{equation}
Figure \ref{fig2} gives the corresponding Feynman diagram.
\begin{figure}
    \centering
    \includegraphics[scale=0.56]{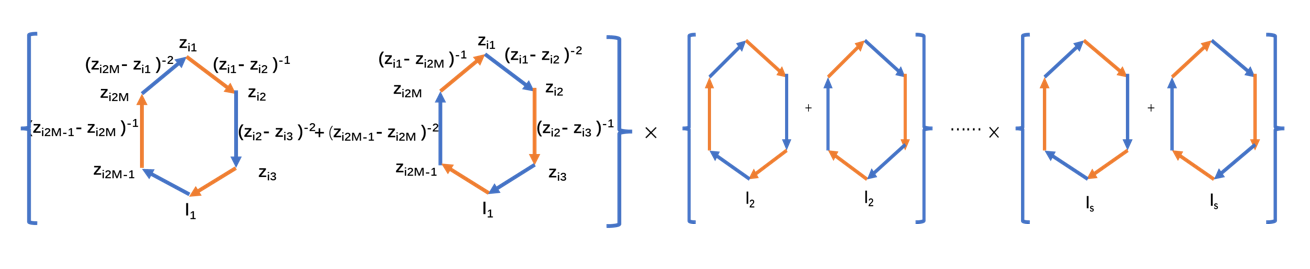}
    \caption{Orange line represents contraction between fields $b,c$, blue line represents contraction between fields $\partial\beta,\gamma$. $l_{i}$ labels a loop. For each loop, we add two graphs because they have the same underlying graph but different propagators. 
   By Wick's theorem, summing over all the possible products of loops gives the $2n$ point amplitude of supercurrent operators.}
    \label{fig2}
\end{figure}\\

The amplitude associated with each Feynman diagram is
\begin{equation}
\text{sign}\times g_{l_{1}}(z_{1_{1}}z_{1_{2}}\cdots z_{1_{2m_{1}}})\times g_{l_{2}}\times\cdots\times g_{l_{s}}.\\
\end{equation}
where sign is the signature of 
$\begin{pmatrix}
    1&2&\cdots&2n\\
    l_{1}&l_{2}&\cdots&l_{s}\\
\end{pmatrix}$
which is a reflection of anti-communitivity of $bc$ fields. And 
\begin{equation}
\begin{split}
&g_{l_i}(z_{i_{1}}\cdots z_{i_{2m_{i}}})=\frac{N}{(z_{i_{1}}-z_{i_{2}})(z_{i_{2}}-z_{i_{3}})^{2}\cdots(z_{i_{2m_{i}-1}}-z_{i_{2m_{i}}})(z_{i_{2m_{i}}}-z_{i_{1}})^{2}}\\
&-\frac{N}{(z_{i_{1}}-z_{i_{2}})^{2}(z_{i_{2}}-z_{i_{3}})\cdots(z_{i_{2m_{i}-1}}-z_{i_{2m_{i}}})^{2}(z_{i_{2m_{i}}}-z_{i_{1}})}\label{funcg}.
\end{split}
\end{equation}
where the factor of $N$ is a result of summing over all ghost fields. Since $c=3N$ in this free field realization, $\la G(z_{1})G(z_{2})...G(z_{2n})\ra_{c=3N}=H(c=3N,z_{1},z_{2},...z_{2n})$ holds for infitely many $N\in\mathbb{Z}_{>0}$, thus the theorem holds by the \textbf{Lemma} \ref{lemma: identical rational functions}.\\
\end{proof}
Below we give an equivalent way of writing the $2n$ points correlators of supercurrent:
\begin{theorem}
    \begin{equation}
\begin{aligned}
\la G(z_{1})\cdots G(z_{2n})\ra=\mathop{\sum \limits_{I_{1}\cup I_{2}\cdots\cup I_{s}=(1 2 \cdots 2n)}}\limits_{\text{cluster\,decomposition}}\left(\frac{2c}{3}\right)^{s}\times \mathrm{sign}\begin{pmatrix}
    1&2&\cdots&2n\\
    I_{1}&I_{2}&\cdots&I_{s}\\
\end{pmatrix}
\mathop{\prod}\limits_{i=1}^{s}\hat{g}_{I_{i}}\label{e1}
\end{aligned}
\end{equation}
where each cluster has  $|I_{i}|=2M_{i}$ elements, $n=M_{1}+M_{2}+\cdots+M_{s}$. Moreover, we have
\begin{equation}
\begin{aligned}
&\hat{g}_{I_{i}}(z_{i_{1}},z_{i_{2}},\cdots,z_{i_{2M_{i}}})\\
&=\frac{1}{2M_{i}}\mathop{\sum}\limits_{\sigma\in S_{2M_{i}}}\frac{\mathrm{sign}(\sigma)}{(z_{\sigma(i_{1})}-z_{\sigma(i_{2})})(z_{\sigma(i_{2})}-z_{\sigma(i_{3})})^{2}\cdots(z_{\sigma(i_{2M_{i}-1})}-z_{\sigma(i_{2M_{i}})})(z_{\sigma(i_{2M_{i}})}-z_{\sigma(i_{1})})^{2}},
\end{aligned}
\end{equation}
where $(z_{i_{1}},z_{i_{2}},...,z_{i_{2M_{i}}})\in I_{i}$.
\end{theorem}

\begin{proof}
This is a direct result from Eq (\ref{funcg}). Note that the first and second term of Eq (\ref{funcg}) are related via a  permutation $\sigma$ from $(i_{1},i_{2},\cdots,i_{2M_{i}})$ to $(i_{2},i_{3},\cdots,i_{2M_{i}},i_{1})$.\\
Every $I_{i}$ is ordered compared with $l_{i}$, for example: if $l_{i}=(3,1,2,4)$, then $I_{i}=(1,2,3,4)$. Their sign are related:
\begin{equation}
    \mathrm{sign}\begin{pmatrix}
    1&2&\cdots&2n\\
    l_{1}&l_{2}&\cdots&l_{s}\\
\end{pmatrix}=\mathrm{sign}\begin{pmatrix}
    1&2&\cdots&2n\\
    I_{1}&I_{2}&\cdots&I_{s} \\
\end{pmatrix}\times\mathop{\mathop{\prod}\limits_{\sigma(i)\in S_{2M_{i}}}}\limits_{1\leq i\leq s}\text{sign}(\sigma(i))
\end{equation}
The factor of $\frac{1}{2M_{i}}$ is because that the Feynman diagram corresponding to each cluster $I_{i}$ is invariant under $M_{i}$ elements in permutation group $S_{2M_{i}}$, where each element $\sigma_{j=1,2,...,M_{i}}$ is:
\begin{equation}
\begin{aligned}
    &\sigma_{j}(z_{k})=z_{2M_{i}-2j+k}, \quad\text{for}\quad k=1,2,...2j\\
    &\sigma_{j}(z_{k})=z_{k-2j}, \quad\text{for} \quad k=2j+1,...,2M_{i}.\\
\end{aligned}
\end{equation}
\end{proof}
Note here that for supercurrent, the number of vertices in each cluster $|I_{i}|$($i=1,2,...,s$) must be even. However, for stress-energy tensor, $|I_{i}|$ could be even or odd.
\subsection{Clustering behaviours of the wavefunction}
From the results in previous section, we now explore how the structure of supercurrent correlation functions affect clustering behaviours of their fractional quantum Hall wavefunctions. Without loss of generality, let us set $p=3,~ q=1$ for $U(1)$ charge sector such that wavefunctions are symmetric polynomials.  
\begin{corollary}\label{cor: clustering}
Let $1\leq k\leq n$, then 
\begin{equation}
\begin{split}
    \lim_{z_{2i-1}\rightarrow z_{2i}}\mathop{\prod}\limits_{i=1}^{k}(z_{2i-1}-z_{2i})^{3}\la G(z_{1})\cdots G(z_{2n})\ra=\left(\frac{2c}{3}\right)^{k}\la G(z_{2k+1})\cdots G(z_{2n})\ra\label{clust}
    \end{split}
\end{equation}
\end{corollary}

\begin{proof}
If $z_{2i-1}$ and $z_{2i}$ are not in the same cluster, then Eq (\ref{clust}) will vanish when we take the limit $z_{2i-1}\rightarrow z_{2i}$. If $z_{2i-1}$ and $z_{2i}$ are in the same cluster $I_{i}$, and the length of $I_{i}$ is more than 2 (i.e, at least 4 vertices in this group), the the terms related to $(z_{2i-1}-z_{2i})$ are either proportional to $\frac{1}{z_{2i-1}-z_{2i}}$ or $\quad\frac{1}{(z_{2i-1}-z_{2i})^{2}}$,
which also vanishes in the limit when $z_{2i-1}\rightarrow z_{2i}$. The only non zero contribution is if there exist a $I_{i}=(z_{2i-1},z_{2i})$, in this case:
\begin{equation}
    (z_{2i-1}-z_{2i})^{3}g_{I_{i}}(z_{2i-1},z_{2i})=(z_{2i-1}-z_{2i})^{3}\frac{\frac{2c}{3}}{(z_{2i-1}-z_{2i})^{3}}=\frac{2c}{3},
\end{equation}
 Therefore, the only non zero term in this limit is the following cluster decomposition:
\begin{equation}
    I_{1}\cup I_{2}\cup\cdots\cup I_{k}\cup I_{k+1}\cup\cdots\cup I_{s},
\end{equation}
where for $1\leq i\leq k$, $I_{i}=(z_{\sigma(2i-1)},z_{\sigma(2i)})$, $\sigma\in S_{2k}$. Then the conlcusion follows from the definition of $\hat{g}_I$.
\end{proof}

\begin{corollary}
Let 
\begin{equation}
\phi(z_{1},\cdots,z_{2n})=\mathop{\prod}\limits_{1\leq i<j\leq 2n}(z_{i}-z_{j})^{3}\la G(z_{1})\cdots G(z_{2n})\ra\label{e2},
\end{equation}
then \begin{equation}
\phi=\mathop{\sum}\limits_{s=1}^{n}\left(\frac{2c}{3}\right)^{s}Q_{2n}^{s}(z_{1},\cdots,z_{2n})\label{def},
\end{equation}
where $Q^s_{2n}(z_1,\cdots,z_{2n})$ are symmetric polynomials determined by
\begin{equation}
\begin{aligned}
    Q_{2n}^{1}&=\frac{1}{2n}\mathop{\sum}\limits_{\sigma\in S_{2n}}\frac{\mathop{\prod}\limits_{1\leq i<j\leq 2n}(z_{\sigma(i)}-z_{\sigma(j)})^{3}}{(z_{\sigma(1)}-z_{\sigma(2)})(z_{\sigma(2)}-z_{\sigma(3)})^{2}\cdots(z_{\sigma(2n-1)}-z_{\sigma(2n)})(z_{\sigma(2n)}-z_{\sigma(1)})^{2}}\label{Q1},\\
\end{aligned}
\end{equation}
and 
\begin{equation}
        Q_{2n}^{s}=\mathop{\mathop{\sum}\limits_{I_{1}\cup I_{2}\cdots\cup I_{s}=(1 2\cdots 2n)}}\limits_{\text{cluster\,decomposition}, |I_{i}|=2m_{i}}\mathop{\prod}\limits_{i=1}^{s} Q_{2m_{i}}^{1}\mathop{\prod}\limits_{\substack{j\in I_{r},k\in I_{t}\\r<t}}(z_{j}-z_{k})^{3}.
        \label{R}
\end{equation}
\end{corollary}
\noindent 
\begin{proof}
   
Following Eq (\ref{e1}), Eq (\ref{e2}) and Eq (\ref{def}), we have:
\begin{equation}
    \mathop{\sum}\limits_{s=1}^{n}\left(\frac{2c}{3}\right)^{s}Q_{2n}^{s}(z_{1},\cdots,z_{2n})=\mathop{\prod}\limits_{1\leq i<j\leq 2n}(z_{i}-z_{j})^{3}\mathop{\sum \limits_{I_{1}\cup I_{2}\cdots\cup I_{s}=(1 2 \cdots 2n)}}\limits_{\text{cluster\,decomposition}}\left(\frac{2c}{3}\right)^{s}\times \mathrm{sign}\begin{pmatrix}
    1&2&\cdots&2n\\
    I_{1}&I_{2}&\cdots&I_{s}\\
\end{pmatrix}
\mathop{\prod}\limits_{i=1}^{s}\hat{g}_{I_{i}}.
\end{equation}
Comparing the polynomial expansion of $\phi$ and G gives Eq (\ref{Q1}) and Eq (\ref{R}). \\For s=1, it gives:
\begin{equation}
\begin{aligned}
    Q_{2n}^{1}&=\mathop{\prod}\limits_{1\leq i<j\leq 2n}(z_{i}-z_{j})^{3}\hat{g}_{I_{1}}\\
    &=\frac{1}{2n}\mathop{\sum}\limits_{\sigma\in S_{2n}}\frac{\mathrm{sign}(\sigma)\mathop{\prod}\limits_{1\leq i<j\leq 2n}(z_{i}-z_{j})^{3}}{(z_{\sigma(i_{1})}-z_{\sigma(i_{2})})(z_{\sigma(i_{2})}-z_{\sigma(i_{3})})^{2}\cdots(z_{\sigma(i_{2n-1})}-z_{\sigma(i_{2n})})(z_{\sigma(i_{2n})}-z_{\sigma(i_{1})})^{2}}\\
    &=\frac{1}{2n}\mathop{\sum}\limits_{\sigma\in S_{2n}}\frac{\mathop{\prod}\limits_{1\leq i<j\leq 2n}(z_{\sigma(i)}-z_{\sigma(j)})^{3}}{(z_{\sigma(i_{1})}-z_{\sigma(i_{2})})(z_{\sigma(i_{2})}-z_{\sigma(i_{3})})^{2}\cdots(z_{\sigma(i_{2n-1})}-z_{\sigma(i_{2n})})(z_{\sigma(i_{2n})}-z_{\sigma(i_{1})})^{2}},\\
\end{aligned}
\end{equation}where we use the identity:
\begin{equation}
    \mathop{\prod}\limits_{1\leq i<j\leq 2n}(z_{\sigma(i)}-z_{\sigma(j)})^{3}=\text{sign}(\sigma)\mathop{\prod}\limits_{1\leq i<j\leq 2n}(z_{i}-z_{j})^{3}.
\end{equation}
For $2\leq s\leq n$, we have:
\begin{equation}
\begin{aligned}
    Q_{2n}^{s}&=\mathop{\prod}\limits_{1\leq i<j\leq 2n}(z_{i}-z_{j})^{3}\mathop{\mathop{\sum}\limits_{I_{1}\cup I_{2}\cdots\cup I_{s}=(1 2\cdots 2n)}}\limits_{\text{cluster\,decomposition}, |I_{i}|=2m_{i}}\mathrm{sign}\begin{pmatrix}
    1&2&\cdots&2n\\
    I_{1}&I_{2}&\cdots&I_{s}\\
\end{pmatrix}
\mathop{\prod}\limits_{i=1}^{s}\hat{g}_{I_{i}}\\
&=\mathop{\mathop{\sum}\limits_{I_{1}\cup I_{2}\cdots\cup I_{s}=(1 2\cdots 2n)}}\limits_{\text{cluster\,decomposition}, |I_{i}|=2m_{i}}\mathop{\prod}\limits_{i=1}^{s} Q_{2m_{i}}^{1}\mathop{\prod}\limits_{\substack{j\in I_{r},k\in I_{t}\\ 1\leq r<t\leq s}}(z_{j}-z_{k})^{3},
\end{aligned}
\end{equation}

where the following identity is used:
\begin{equation}
\begin{aligned}
  \mathop{\prod}\limits_{1\leq i<j\leq 2n}(z_{i}-z_{j})^{3}&= \mathop{\mathop{\prod}\limits_{ a_{i}<b_{i}z_{a_{i}},z_{b_{i}}\in I_{i}}}\limits_{1\leq i\leq s}(z_{a_{i}}-z_{b_{i}})^{3}\mathop{\prod}\limits_{\substack{j\in I_{r},k\in I_{t}\\j<k}}(z_{j}-z_{k})^{3}\\
  &= \mathop{\mathop{\prod}\limits_{ a_{i}<b_{i}z_{a_{i}},z_{b_{i}}\in I_{i}}}\limits_{1\leq i\leq s}(z_{a_{i}}-z_{b_{i}})^{3}\mathrm{sign}\begin{pmatrix}
    1&2&\cdots&2n\\
    I_{1}&I_{2}&\cdots&I_{s}\\
\end{pmatrix}\mathop{\prod}\limits_{\substack{j\in I_{r},k\in I_{t}\\1\leq r<t\leq s}}(z_{j}-z_{k})^{3}
  \end{aligned}
\end{equation}
and 
\begin{equation}
   Q_{2m_{i}}^{1}=\mathop{\prod}\limits_{ a_{i}<b_{i},a_{i},b_{i}\in I_{i}}(z_{a_{i}}-z_{b_{i}})^{3}\hat{g}_{I_{i}} 
\end{equation}
\end{proof}
 It is easy to deduce from Corollary \ref{cor: clustering} and Eq (\ref{e2}) that
\begin{equation}\label{eqn: clustering}
    \begin{split}
        &\phi(z_1=z_2=Z_1,\cdots,z_{2l-1}=z_{2l}=Z_{l},\cdots)\\
        &=\left(\frac{2c}{3}\right)^{l}\prod_{1\leq j<k\leq l}(Z_j-Z_k)^{12}\prod_{j=1}^{l}\prod_{k=2l+1}^{2n}(Z_j-z_k)^6 \phi(z_{2l+1},\cdots,z_{2n}),
    \end{split}
\end{equation}
then do similar comparasion of polynomial coeeficient of  Eq (\ref{def}) shows that $Q^s_{2n}$ satisfy the following clustering properties:
\begin{equation}
    \begin{split}
        &Q^s_{2n}(z_1=z_2=Z_1,\cdots,z_{2l-1}=z_{2l}=Z_{l},\cdots)\\
&=\prod_{1\leq j<k\leq l}(Z_j-Z_k)^{12}\prod_{j=1}^{l}\prod_{k=2l+1}^{2n}(Z_j-z_k)^6 Q^{s-l}_{2n-2l}(z_{2l+1},\cdots,z_{2n}).
    \end{split}
    \label{Q}
\end{equation}
Following \cite{Estienne_2010, Moore:1991ks}, the clustering properties of wave function $Q_{2n}^{s}$ is a reflection of $\mathbb{Z}^{s}$ symmetry.
Note that Eq (\ref{Q}) is invariant under permutations, namely:
    \begin{equation}
    \begin{split}
        &Q^s_{2n}(z_1=z_2=Z_1,\cdots,z_{2l-1}=z_{2l}=Z_{l},z_{2l+1},...,z_{2n})\\
&=Q^s_{2n}(z_{\sigma(1)}=z_{\sigma(2)}=Z_{1},\cdots,z_{\sigma(2l-1)}=z_{\sigma(2l)}=Z_{l}, z_{\sigma(2l+1)},...,z_{\sigma(2n)}).
    \end{split}
\end{equation}
which corresponds to the clustering decomposition: 
\begin{equation}
    I_{1}\cup I_{2}\cup\cdots\cup I_{l}\cup I_{l+1}\cup\cdots\cup I_{s},
\end{equation}
where for $1\leq j\leq l$, $I_{j}=(z_{\sigma(2j-1)}=Z_{j},z_{\sigma(2j)}=Z_{j})$, $\sigma\in S_{l}$. For $1\leq s\leq n-1$, one can prove that $l\leq s-1$ must hold by using proof by contradiction. If $l=s$, then the clustering decomposition that make $Q_{2n}^{s}$ non zero is:
\begin{equation}
    I_{1}\cup I_{2}\cup\cdots\cup I_{s}=(1,2,...2n),
\end{equation}
the length of each $I_{i}$ is 2, however, since $s\leq n-1$, the total length of this clustering decomposition is $2s\leq2n-2<2n$, therefore for $2\leq s\leq n-1$, $1\leq l\leq s-1$. For $s=n$, the clustering decomposition that make $Q_{2n}^{n}$ non zero is:
\begin{equation}
    I_{1}\cup I_{2}\cup\cdots\cup I_{n}=(1,2,...,2n),
\end{equation}
the length of each $I_{i}$ is 2, so in this case, $1\leq l\leq n$. 
Mathematically, each $Q^s_{2n}(z_1,\cdots, z_{2n})$ is a symmetric polynomial with $2n$ coordinates such that whenever $k+1=3$ particles coincide, the polynomial vanishes as power $r=6$, which is a direct result from Eq (\ref{Q}). Individually, for $1\leq s\leq n-1$, $Q^s_{2n}(z_1,\cdots, z_{2n})$ vanishes whenever $s$ clusters of two particles coincide but not vanishing if only $s-1$ clusters of two particles coincide as shown in Eq.(\ref{Q}). For $s=n$, $Q_{2n}^{n}$ will be non vanishing even for $n$ clusters of two particles coincide.

Below we give some simple examples of these symmetric polynomials:
\begin{example}
Let $z_{ij}=z_i-z_j$, then
\begin{equation}
\begin{aligned}
&Q_{2}^{1}=1,\\
&Q_{4}^{1}=3z_{12}^{2}z_{13}^{2}z_{14}^{2}z_{23}^{2}z_{24}^{2}z_{34}^{2},\\
&Q_{6}^{2}=3\left(\mathop{\prod}\limits_{a<b}^{6}z_{ab}^{2}\right)\left(\mathop{\sum}\limits_{i<j}^{6}\frac{1}{z_{ij}^{2}}\mathop{\prod}\limits_{k\neq i,j}^{6}z_{ik}z_{jk}\right).
\end{aligned}
\end{equation}
\end{example}
\begin{corollary}
$Q_{4}^{1}=3z_{12}^{2}z_{13}^{2}z_{14}^{2}z_{23}^{2}z_{24}^{2}z_{34}^{2}.$
\end{corollary}
\begin{proof}
    Following Eq(\ref{Q1}),we have:
    \begin{equation}
        Q_{4}^{1}=\frac{1}{4}\mathop{\prod}\limits_{i<j}^{4}z_{ij}^{2}\mathop{\sum}\limits_{\sigma}\frac{z_{\sigma(1)\sigma(3)}z_{\sigma(2)\sigma(4)}}{z_{\sigma(1)\sigma(4)}z_{\sigma(2)\sigma(3)}}\label{Q4}
    \end{equation}
For the cross ratio, take $x=\frac{z_{13}z_{24}}{z_{14}z_{23}}$ as an example, one can prove that the cross ratio over the $4!=24$ permutations have 6 distinct value:
\begin{equation}
    x,\quad1-x,\quad\frac{1}{x},\quad1-\frac{1}{x},\quad\frac{1}{1-x},\quad\frac{x}{x-1}
\end{equation}
The stabilizer for the cross ratio $\frac{z_{\sigma(1)\sigma(3)}z_{\sigma(2)\sigma(4)}}{z_{\sigma(1)\sigma(4)}z_{\sigma(2)\sigma(3)}}$ is:
\begin{equation}
\begin{aligned}
     &(\sigma(1)\sigma(2)\sigma(3)\sigma(4)),\quad(\sigma(2)\sigma(1)\sigma(4)\sigma(3))\\
     &(\sigma(4)\sigma(3)\sigma(2)\sigma(1)),\quad(\sigma(3)\sigma(4)\sigma(1)\sigma(2))\\
\end{aligned}
\end{equation}
Therefore:
\begin{equation}
    \begin{aligned}
        &\mathop{\sum}\limits_{\sigma}\frac{z_{\sigma(1)\sigma(3)}z_{\sigma(2)\sigma(4)}}{z_{\sigma(1)\sigma(4)}z_{\sigma(2)\sigma(3)}}\\
        &=4(x+1-x+\frac{1}{x}+1-\frac{1}{x}+\frac{1}{1-x}+\frac{x}{x-1})\\
        &=12,
    \end{aligned}
\end{equation}
which implies that:
\begin{equation}
Q_{4}^{1}=3z_{12}^{2}z_{13}^{2}z_{14}^{2}z_{23}^{2}z_{24}^{2}z_{34}^{2}.
\end{equation}
\end{proof}
\begin{corollary}
    $Q_{2n}^{n}=\mathrm{Pf}\left(z_{ij}^{-3}\right)\mathop{\prod}\limits_{i<j}^{2n}z_{ij}^{3}.\label{Q2n}$
\end{corollary}
\begin{proof}
    For a skew-symmetric matrix A, Pfaffian of A satisfies $\mathrm{Pf}(A)^{2}=det(A)$. Let $A=(a_{ij})$ be a $2n \times 2n$ skew-symmetric matrix, the explicit expression of its Pfaffian is:
    \begin{equation}
        \mathrm{Pf}(A)=\frac{1}{2^{n}n!}\mathop{\sum}\limits_{\sigma\in S_{2n}}\text{sign}(\sigma)\mathop{\prod}\limits_{i=1}^{n}a_{\sigma(2n-1),\sigma(2n)}.
    \end{equation}
   which leads to the expression of  $\mathrm{Pf}\left(z_{ij}^{-3}\right)$:
    \begin{equation}
        \mathrm{Pf}\left(z_{ij}^{-3}\right)=\frac{1}{2^{n}n!}\mathop{\sum}\limits_{\sigma\in S_{2n}}\text{sign}(\sigma)\mathop{\prod}\limits_{i=1}^{n}z_{\sigma(2i-1)\sigma(2i)}^{-3}.
        \label{Pf}
    \end{equation}
Following Eq (\ref{R}), for $Q_{2n}^{n}$, the cluster decomposition is given by:
\begin{equation}
    I_{1}\cup I_{2}\cup\cdots\cup I_{n}=(1,2,...,2n)
\end{equation}
Since the length of each $I_{i}$ is at least 2 and the total length is $2n$. We have for $1\leq i\leq n$, $|I_{i}|=2$. Note that $Q_{2}^{1}=1$, by Eq (\ref{R}) we have:
\begin{equation}
\begin{aligned}
    Q_{2n}^{n}&=\mathop{\mathop{\sum}\limits_{I_{1}\cup I_{2}\cdots\cup I_{n}=(1 2\cdots 2n)}}\limits_{\text{cluster\,decomposition}, |I_{i}|=2}\mathop{\prod}\limits_{i=1}^{n} Q_{I_{i}}^{1}\mathop{\prod}\limits_{\substack{j\in I_{r},k\in I_{t}\\r<t}}(z_{j}-z_{k})^{3}\\
    &=\mathop{\mathop{\sum}\limits_{I_{1}\cup I_{2}\cdots\cup I_{n}=(1 2\cdots 2n)}}\limits_{\text{cluster\,decomposition}, |I_{i}|=2}\mathop{\prod}\limits_{\substack{j\in I_{r},k\in I_{t}\\r<t}}(z_{j}-z_{k})^{3}.\\
    &=\mathop{\mathop{\sum}\limits_{I_{1}\cup I_{2}\cdots\cup I_{n}=(1 2\cdots 2n)}}\limits_{\text{cluster\,decomposition}, |I_{i}|=2}\mathop{\prod}\limits_{i<j}^{2n}z_{ij}^{3}\mathop{\prod}\limits_{\substack{l,k\in I_{r}\\r=1,2,...n}}(z_{lk})^{-3}\\
    &=\left(\mathop{\mathop{\sum}\limits_{I_{1}\cup I_{2}\cdots\cup I_{n}=(1 2\cdots 2n)}}\limits_{\text{cluster\,decomposition}, |I_{i}|=2}\mathop{\prod}\limits_{\substack{l,k\in I_{r}\\r=1,2,...n}}(z_{lk})^{-3}\right)\mathop{\prod}\limits_{i<j}^{2n}z_{ij}^{3}\\
    \end{aligned}
    \label{clu}
\end{equation}
Following Eq (\ref{Pf}), there is a subgroup of $S_{2n}$ which leaves the expression $\text{sign}(\sigma)\mathop{\prod}\limits_{i=1}^{n}z_{\sigma(2i-1)\sigma(2i)}^{-3}$ invariant. For example, if we take the cluster decomposition to be:
\begin{equation}
    I_{i}=(z_{2i-1},z_{2i})\quad i=1,2,...,n
\end{equation}
The following equation holds for a subgroup of $S_{2n}$:
\begin{equation}
    \text{sign}(\sigma)\mathop{\prod}\limits_{i=1}^{n}z_{\sigma(2i-1)\sigma(2i)}^{-3}=\mathop{\prod}\limits_{i=1}^{n}z_{(2i-1)(2i)}^{-3},
\end{equation}
where the generators of this subgroup are given by:
\begin{equation}
\begin{aligned}
    &a_{i}=\sigma(z_{2i-1},z_{2i})=(z_{2i},z_{2i-1})\quad i=1,2,...,n\\
    &b_{j}=\sigma(z_{1},z_{2},z_{2j-1},z_{2j})=(z_{2j-1},z_{2j},...,z_{1},z_{2})\quad j=1,2,...,n-1,\\
\end{aligned}
\end{equation}
One can check that the order of this subgroup is $2^{n}n!$. Therefore:
\begin{equation}
\begin{aligned}
    &\mathrm{Pf}\left(z_{ij}^{-3}\right)=\frac{1}{2^{n}n!}\mathop{\sum}\limits_{\sigma\in S_{2n}}\text{sign}(\sigma)\mathop{\prod}\limits_{i=1}^{n}z_{\sigma(2i-1)\sigma(2i)}^{-3}\\
    &=\left(\mathop{\mathop{\sum}\limits_{I_{1}\cup I_{2}\cdots\cup I_{n}=(1 2\cdots 2n)}}\limits_{\text{cluster\,decomposition}, |I_{i}|=2}\mathop{\prod}\limits_{\substack{l,k\in I_{r}\\r=1,2,...n}}(z_{lk})^{-3}\right).\\
\end{aligned}
\end{equation}
thus the Corollary holds by Eq (\ref{clu}).
\end{proof}
There is an equivalent representation of the ground state wave function if we expand it with another set of symmetric polynomials with same clustering properties. A nice feature for choosing this new set of symmetric polynomials is that it makes direct connections with our limiting case when $m\rightarrow \infty$ and $c\rightarrow \frac{3}{2}$. One can check that, for the following symmmetric polynomials:
\begin{equation}
    P_{2n}^{s}(z_1,\cdots,z_{2n})=\frac{1}{3}\mathop{\sum}\limits_{j=1}^{s}Q_{2n}^{j}(z_1,\cdots,z_{2n})\quad s=1,2,...n-1,
\end{equation}
\begin{equation}
    P_{2n}^{n}(z_1,\cdots,z_{2n})=\frac{1}{3}\mathrm{Pf}^{3}\left(z_{ij}^{-1}\right)\mathop{\prod}\limits_{i<j}^{2n}z_{ij}^{3}.
\end{equation}
We have:
\begin{corollary}
\begin{align}
    \phi(z_1,\cdots, z_{2n})=\left(\frac{2c}{3}\right)^{n}3P_{2n}^{n}(z_1,\cdots,z_{2n})+\sum_{s=1}^{n-1}\left(\frac{2c}{3}\right)^{s}(3-2c)P^s_{2n}(z_1,\cdots,z_{2n}),
\end{align}
where $P^s_{2n}$ share the same clustering properties with $Q^s_{2n}$, namely, 
\begin{equation}
    \begin{split}
        &P^s_{2n}(z_1=z_2=Z_1,\cdots,z_{2l-1}=z_{2l}=Z_{l},\cdots)\\
&=\prod_{1\leq j<k\leq l}(Z_j-Z_k)^{12}\prod_{j=1}^{l}\prod_{k=2l+1}^{2n}(Z_j-z_k)^6 P^{s-l}_{2n-2l}(z_{2l+1},\cdots,z_{2n}).
    \end{split}
\end{equation}
\end{corollary}
\begin{proof}
First of all, we show that $\phi(z_1,\cdots, z_{2n})=\mathrm{Pf}^{3}\left(z_{ij}^{-1}\right)\mathop{\prod}\limits_{i<j}^{2n}z^{3}_{ij}$ when $c=3/2$. Consider a CFT of three Majorana fermions $\psi_1,\psi_2,\psi_3$ with OPEs
\begin{align}
    \psi_a(z)\psi_b(w)=\frac{\delta_{ab}}{z-w}+\mathcal O(1).
\end{align}
This system has $\mathcal N=1$ superconformal symmetry with $c=3/2$, namely
\begin{align}
    T(z)=\frac{1}{2}\sum_{a=1}^3:\partial\psi_a(z)\psi_a(z):\quad G(z)=i\psi_1(z)\psi_2(z)\psi_3(z)
\end{align}
satisfy the $\mathcal N=1$ super Virasoro algebra. Since $\psi_1,\psi_2,\psi_3$ are independent, the correlator $\la G(z_1)\cdots G(z_{2n})\ra$ is simply the product of three correlators $\la \psi_1(z_1)\cdots \psi_1(z_{2n})\ra\la \psi_2(z_1)\cdots \psi_2(z_{2n})\ra\la \psi_3(z_1)\cdots \psi_3(z_{2n})\ra$, which can be easily computed:
\begin{align}
    \la \psi_a(z_1)\cdots \psi_a(z_{2n})\ra=\mathrm{Pf}\left(z_{ij}^{-1}\right),\; a=1,2,3.
\end{align}
Thus $\phi(z_1,\cdots, z_{2n})=\mathrm{Pf}^{3}\left(z_{ij}^{-1}\right)\mathop{\prod}\limits_{i<j}^{2n}z^{3}_{ij}$ when $c=3/2$. Finally, the clustering properties follow from the equation \eqref{eqn: clustering} and the definition of $P^{s}_{2n}$.
\end{proof}
The dominant clustering behaviours of our wavefunctions for any $2n$ particles, vanishing whenever three particles coincide, is fixed by $\mathrm{Pf}^{3}\left(z_{ij}^{-1}\right)\mathop{\prod}\limits_{i<j}^{2n}z^{3}_{ij}$ which is associated to the three Majorana fermion unitary CFT. This is in contrast to previous cases studied where Jack polynomials associated with non-unitary CFTs plus ``healing'' polynomials are used for generating wavefunctions of unitary CFTs.

\section{Discussion}
In this paper, we achieved two things in our studies of $\mathcal{N}=1$ SCFTs. We found explicit characterizations of our $\mathcal{N}=1$ theories in terms of a parafermion theory, an Ising theory and a free boson theory by first pairing our $\mathcal{N}=1$ theory with a parafermion theory. Supercurrent operator $G$ can also be written as linear combinations of operators from its constituent CFTs. By utilizing free field methods, we worked out explicit ground state wavefunctions of fractional quantum Hall states based on our $\mathcal{N}=1$ theories. We have also shown clustering properties of these ground state wavefunctions. Several questions remain. In our case, even though we have worked out an explicit decomposition of the supercurrent operator $G$ into parafermion operators, Majorana fermion operators and $U(1)$ vertex operators, we do not immediately see advantage of this decomposition in calculating correlation functions of $G$. Our preliminary calculations do not show much simplification from operator decomposition approach. On the other hand, we believe decomposition worked out here should be applicable to boundary conforma field theories with $\mathcal{N}=1$ superconformal algebra \cite{chen1,Makabe_2017,Richard_2002,Nepomechie_2001}. Speaking of wavefunctions, previous studies of Jack polynomials and Read-Rezayi series of fractional quantum Hall states \cite{Estienne2010Clustering, Estienne2009Relating, Estienne2012SpinSinglet, Estienne2012Conformal,Read_2009,Read:1998ed,Bernevig_2008,Bernevig_2009,Estienne_2010,Simon_2007,Simon_2009} have demonstarted that distinct clustering behaviours indicate distinct topological orders. In our case, at least for ground state, wavefunctions are the same in terms of their clustering behaviours. The only changing parameter is central charge $c$. 

We would like to clarify the role of clustering in diagnosing topological order.
Clustering conditions are local vanishing constraints on the ground-state polynomial wavefunction.
In many prominent families (e.g. certain Laughlin/Read--Rezayi constructions), distinct clustering patterns
often correlate with distinct topological orders; however, clustering by itself is generally not a complete
invariant of topological order, and different topological orders can share the same minimal local
clustering constraints while differing in other topological data.

In our construction, the leading clustering behaviour of the ground state is dictated by the universal
$\mathcal{N}=1$ super-Virasoro OPE of the supercurrent $G$ and therefore persists across the family.
What changes between different members is encoded in the central charge (and, correspondingly, in the
relative weights among the different clustering sectors, e.g. the coefficients in Eq.~(3.69)),
rather than in the local vanishing order itself.
To distinguish the associated topological orders more sharply, one should incorporate additional data,
such as the chiral central charge and, crucially, the quasiparticle sectors.
In the CFT construction, this corresponds to studying conformal blocks with insertions of other primary
fields (for example Ramond spin fields and/or parafermion primaries) and analyzing their fusion rules
and monodromies, from which topological spins and braiding properties can be extracted.
Our method based on rationality in $c$ and free-field realizations can in principle be extended to such
correlators, and we leave a detailed exploration to future work.

There are two immediate questions: 1. How to understand different topological orders sharing the same clustering behaviours in their wavefunctions? Original studies relate different pseudo-potential \cite{Haldane1983} or Kivelson-Trugman type potential Hamiltonians \cite{Trugman} with Jack polynomial ground state wavefunctions with different clustering behaviours.   2. In thermodynamic limit in which $2n\rightarrow\infty$, each ground state wavefunction should be orthogonal to each other since they represent different topological orders. How to show or give evidences to such tendency? Finally, for all other fractional spin ($\Delta=\frac{N+1}{N})$ generalization of the minimal models, can we always use free field methods to work out their correlation functions?

\emph{Acknowledgments.} We thank for the helpful discussions with Paul Fendley, Steven Simon, Taro Kimura, Davide Gaiotto, Matthias Gaberdiel, Nicolas Regnault, Benoit Estienne, Biao Lian, Prashant Kumar, Jie Wang, Ching Hung Lam and Chongying Dong. We also want to thank Jeffrey Teo for his early collaboration, many helpful discussions and comments on the draft. SN wants to specially thank Joseph Conlon for his help and encouragement during the work. This work is supported by the Alfred P. Sloan Foundation, and NSF through the Princeton University’s Materials Research Science and Engineering Center DMR-2011750. Additional support was provided by the Gordon and Betty Moore Foundation through Grant GBMF8685 towards the Princeton theory program. SN wants to acknowledge funding support from the China Scholarship
Council-FaZheng Group- University of Oxford. Kavli Institute for the Physics and Mathematics of the Universe is supported by World Premier International Research Center Initiative (WPI), MEXT, Japan. YZ would like to thank Perimeter Institute for Theoretical Physics, where part of YZ's work was done as a graduate student there.

\appendix
\section{$S_3$ minimal model}
$S_3$ minimal models have central charge $c=2\left(1-\frac{12}{(m-2)(m+2)}\right)$ with $m=5,6,\cdots$. It has extended algebra beyond Virosora algebra which is generated by
\begin{equation}
\begin{aligned}
T(z)T(w)=&\frac{1}{\left(z-w\right)^4}\left\{\frac{c}{2}+2\left(z-w\right)^2T(w)+\left(z-w\right)^3\partial T(w)+\dots \right\}\\
T(z)G^\pm(w)=&\frac{1}{\left(z-w\right)^2}\left\{\frac{4}{3}G^\pm(w)+(z-w)\partial G^\pm(w)+\dots\right\},\\
G^+(z)G^+(w)=&\frac{\lambda^+}{\left(z-w\right)^{4/3}}\left\{G^-(w)+\frac{1}{2}(z-w)\partial G^-(w)+\dots\right\},\\
G^-(z)G^-(w)=&\frac{\lambda^-}{\left(z-w\right)^{4/3}}\left\{G^+(w)+\frac{1}{2}(z-w)\partial G^+(w)+\dots\right\},\\
G^+(z)G^-(w)=&\frac{1}{\left(z-w\right)^{8/3}}\left\{\frac{3c}{8}+\left(z-w\right)^2T(w)+\dots\right\}.
\end{aligned}
\end{equation} where $G^\pm$ are operators with scaling dimension $\Delta=\frac{4}{3}$. This series of conformal field theories include examples such as the $\mathbb{Z}_6$ parafermion CFT. Each member has a known coset construction
\beq
\frac{SU(2)_4\times SU(2)_{m-4}}{SU(2)_m}
\eeq
 This series of cosets has the following relation
\begin{equation}
\frac{SU(2)_m}{U(1)_{2m}}\times \frac{SU(2)_4\times SU(2)_{m-4}}{SU(2)_m}=\frac{SU(2)_{m-4}}{U(1)_{2(m-4)}}\times \frac{SU(2)_4}{U(1)_8}\times U(1)_{8m(m-4)}
\end{equation} at the level of stress-energy tensor and $\frac{SU(2)_m}{U(1)_{2m}}$ is the coset for $\mathbb{Z}_m$ parafermion. So this means we can combine a $\mathbb{Z}_m$ parafermion theory with the $m$th member of the $S_3$ minimial model theory and turn them into a combination of a $\mathbb{Z}_{m-4}$ parafermion theory together with a $U(1)$ theory and a $\mathbb{Z}_{4}$ parafermion theory.

Levels for $U(1)$ theories need a little bit more explanations. Let us reshuffle $U(1)$ theories as following:
\beq
\text{Left}\quad (U(1)_{2m})^{-1}\times U(1)_{2(m-4)}\times U(1)_{8}, \quad\text{Right} \quad U(1)_{8m(m-4)}.
\label{U(1)S3}
\eeq To establish an equivalence relation between the left and right hand side, we first write out their Lagrangian densities:
\beq
\mathcal{L}_{l}=\frac{1}{4\pi} \sum_{I,J=1}^{3}K_{l}^{IJ} \partial_x \phi_{I} (\partial_x+\partial_t) \phi_{J}, \quad \mathcal{L}_{r}=\frac{1}{4\pi} K_{r} \partial_x \phi_1 (\partial_x+\partial_t) \phi_1
\eeq where $K_l=\begin{pmatrix}
    -2m & 0 &0\\
    0 & 2(m-4) & 0\\
    0 & 0 & 8
\end{pmatrix}$ and $K_r=8m(m-4)$. It is easy to see that there exists a matrix $M$ with integer entries
\beq
M=\begin{pmatrix}
4-m&-m&0\\
1&1&1\\
-1&-1&1\\
\end{pmatrix}, \quad \det{M}=8>0
\eeq  such that
\beq
M^T
(\begin{pmatrix}
-2m & 0 & 0\\
0 & 2(m-4) & 0\\
0 & 0 & 8\\
\end{pmatrix})
M=\begin{pmatrix}
8m(m-4)\\
\end{pmatrix}\bigoplus 16\sigma_z.
\label{K-matrix}
\eeq 
The right hand side of Eq.~(\ref{K-matrix}) is equivalent to $\mathcal{L}_r$ as we can add a backscattering term to the two counter-propagating modes
\beq
2\cos{(2(\phi_2-\phi_3))}
\eeq to gap out $16\sigma_z$ degrees of freedom ($\phi_2$ and $\phi_3$)\cite{leonid,Kane1992}. We have established an equivalence relation in Eq. (\ref{U(1)S3}).


\section{Wavefunction comparison}
In \cite{Simon_2009}, the author gives a formula of the $2n$ point correlators of $\mathcal N=1$ superconformal currents. In this section, we compare with his results and give a proof on the equivalence for $n=2$ and $n=3$.
\begin{equation}
    \phi_{2n}=\frac{(\frac{c}{3})^{\frac{n}{2}(3-n)}}{n!}\mathop{\sum}\limits_{\sigma\in S_{2n}}\mathop{\prod}\limits_{1\leq r<s\leq n}\chi(z_{\sigma(2r-1)},z_{\sigma(2r)};z_{\sigma(2s-1)},z_{\sigma(2s)}),
\end{equation}
where the function $\chi$ is:
\begin{equation}
\chi(z_{1},z_{2};z_{3},z_{4})=z_{13}^{3}z_{14}^{3}z_{23}^{3}z_{24}^{3}\left(\frac{c}{3}+\frac{z_{12}z_{34}}{z_{14}z_{23}}\right).
\end{equation}
For $n=2$, we have:
\begin{equation}
    \begin{aligned}
    \phi_{4}&=\frac{c}{6}\mathop{\sum}\limits_{P\in S_{4}}\chi(z_{P(1)},z_{P(2)};z_{P(3)},z_{P(4)})\\
    &=\frac{c}{6}\mathop{\sum}\limits_{\sigma\in S_{4}}z_{\sigma(1)\sigma(3)}^{3}z_{\sigma(1)\sigma(4)}^{3}z_{\sigma(2)\sigma(3)}^{3}z_{\sigma(2)\sigma(4)}^{3}\left(\frac{c}{3}+\frac{z_{\sigma(1)\sigma(2)}z_{\sigma(3)\sigma(4)}}{z_{\sigma(1)\sigma(4)}z_{\sigma(2)\sigma(3)}}\right).
\end{aligned}
\end{equation}
Note that:
\begin{equation}
\begin{aligned}
&z_{\sigma(1)\sigma(3)}^{3}z_{\sigma(1)\sigma(4)}^{3}z_{\sigma(2)\sigma(3)}^{3}z_{\sigma(2)\sigma(4)}^{3}\frac{z_{\sigma(1)\sigma(2)}z_{\sigma(3)\sigma(4)}}{z_{\sigma(1)\sigma(4)}z_{\sigma(2)\sigma(3)}}\\
&=\mathop{\prod}\limits_{1\leq i<j\leq 4}z_{ij}^{2}\frac{z_{\sigma(1)\sigma(3)}z_{\sigma(2)\sigma(4)}}{z_{\sigma(1)\sigma(2)}z_{\sigma(3)\sigma(4)}}
\end{aligned}
\end{equation}
The sum is over all elements in $S_{4}$. $\mathop{\prod}\limits_{1\leq i<j\leq 4}z_{ij}^{2}$ is invariant under permutation. If we permute $\sigma(2)$ and $\sigma(4)$ with each other, the expression reads:
\begin{equation}
\begin{aligned}
    &\frac{c}{6}\mathop{\prod}\limits_{1\leq i<j\leq 4}z_{ij}^{2}\frac{z_{\sigma(1)\sigma(3)}z_{\sigma(2)\sigma(4)}}{z_{\sigma(1)\sigma(4)}z_{\sigma(2)\sigma(3)}}\\
    &=\frac{2c}{3}Q_{4}^{1}
\end{aligned} 
\end{equation}
where we use Eq (\ref{Q4}).\\
One can also check that:
\begin{equation}
\begin{aligned}
     &\frac{c^{2}}{18}\mathop{\sum}\limits_{\sigma\in S_{4}}z_{\sigma(1)\sigma(3)}^{3}z_{\sigma(1)\sigma(4)}^{3}z_{\sigma(2)\sigma(3)}^{3}z_{\sigma(2)\sigma(4)}^{3}\\
     &=\frac{c^{2}}{18}\times8\times(z_{13}^{3}z_{14}^{3}z_{23}^{3}z_{24}^{3}-z_{12}^{3}z_{14}^{3}z_{23}^{3}z_{34}^{3}+z_{12}^{3}z_{13}^{3}z_{24}^{3}z_{34}^{3})\\
     &=(\frac{2c}{3})^{2}Q_{4}^{2},
\end{aligned}
\end{equation}
where we use Corollary $\ref{Q2n}$.
So $\phi_{4}(z_{1},z_{2},z_{3},z_{4})=(\frac{2c}{3})^{2}Q_{4}^{2}+\frac{2c}{3}Q_{4}^{1}$.\\

For $n=3$, we have:
\begin{equation}
\begin{aligned}
    &\phi_{6}=\frac{1}{6}\mathop{\sum}\limits_{\sigma\in S_{6}}\mathop{\prod}\limits_{1\leq r<s\leq 6}\chi(z_{\sigma(2r-1},z_{\sigma(2r)};z_{\sigma(2s-1)},z_{\sigma(2s)})\\
    &=\frac{1}{6}\mathop{\sum}\limits_{\sigma\in S_{6}}z_{\sigma(1)\sigma(3)}^{3}z_{\sigma(1)\sigma(4)}^{3}z_{\sigma(2)\sigma(3)}^{3}z_{\sigma(2)\sigma(4)}^{3}(\frac{c}{3}+\frac{z_{\sigma(1)\sigma(2)}z_{\sigma(3)\sigma(4)}}{z_{\sigma(1)\sigma(4)}z_{\sigma(2)\sigma(3)}})z_{\sigma(1)\sigma(5)}^{3}z_{\sigma(1)\sigma(6)}^{3}z_{\sigma(2)\sigma(5)}^{3}z_{\sigma(2)\sigma(6)}^{3}\\
    &(\frac{c}{3}+\frac{z_{\sigma(1)\sigma(2)}z_{\sigma(5)\sigma(6)}}{z_{\sigma(1)\sigma(6)}z_{\sigma(2)\sigma(5)}})
    z_{\sigma(3)\sigma(5)}^{3}z_{\sigma(3)\sigma(6)}^{3}z_{\sigma(4)\sigma(5)}^{3}z_{\sigma(4)\sigma(6)}^{3}(\frac{c}{3}+\frac{z_{\sigma(3)\sigma(4)}z_{\sigma(5)\sigma(6)}}{z_{\sigma(3)\sigma(6)}z_{\sigma(4)\sigma(5)}})\\
\end{aligned}
\end{equation}
For term relating $c^{3}$, it is easy to prove that it is equivalent with $Q_{6}^{3}$
.So we start from term that proportional to $c^{2}$:
\begin{equation}
    \begin{aligned}
        &\frac{1}{6}\frac{c^{2}}{9}\mathop{\sum}\limits_{\sigma\in S_{6}}\mathop{\prod}\limits_{1\leq i<j\leq 6}z_{ij}^{2}(\frac{z_{\sigma(1)\sigma(3)}z_{\sigma(2)\sigma(4)}}{z_{\sigma(1)\sigma(2)}z_{\sigma(3)\sigma(4)}}\frac{\mathop{\prod}\limits_{k\neq 5,6}z_{\sigma(k)\sigma(5)}z_{\sigma(k)\sigma(6)}}{z_{\sigma(5)\sigma(6)}^{2}}\\
        &+\frac{z_{\sigma(1)\sigma(5)}z_{\sigma(2)\sigma(6)}}{z_{\sigma(1)\sigma(2)}z_{\sigma(5)\sigma(6)}}\frac{\mathop{\prod}\limits_{k\neq 3,4}z_{\sigma(k)\sigma(3)}z_{\sigma(k)\sigma(4)}}{z_{\sigma(3)\sigma(4)}^{2}}+\frac{z_{\sigma(3)\sigma(5)}z_{\sigma(4)\sigma(6)}}{z_{\sigma(3)\sigma(4)}z_{\sigma(5)\sigma(6)}}\frac{\mathop{\prod}\limits_{k\neq 1,2}z_{\sigma(k)\sigma(1)}z_{\sigma(k)\sigma(2)}}{z_{\sigma(1)\sigma(2)}^{2}})\\
        &=\frac{4c^{2}}{3}\mathop{\prod}\limits_{1\leq a<b\leq 6}z_{ab}^{2}\left(\mathop{\sum}\limits_{i<j}^{6}\frac{1}{z_{ij}^{2}}\mathop{\prod}\limits_{k\neq i,j}^{6}z_{ik}z_{jk}\right)\\
        &=(\frac{2c}{3})^{2}Q_{6}^{2}
    \end{aligned}
\end{equation}
This is because, consider the summation over $S_{6}$ for the following term:
\begin{equation}
    \frac{c^{2}}{54}\mathop{\sum}\limits_{\sigma\in S_{6}}\mathop{\prod}\limits_{1\leq i<j\leq 6}z_{ij}^{2}\frac{z_{\sigma(1)\sigma(3)}z_{\sigma(2)\sigma(4)}}{z_{\sigma(1)\sigma(2)}z_{\sigma(3)\sigma(4)}}\frac{\mathop{\prod}\limits_{k\neq 5,6}z_{\sigma(k)\sigma(5)}z_{\sigma(k)\sigma(6)}}{z_{\sigma(5)\sigma(6)}^{2}}
\end{equation}
For a given $\sigma$, this term is invariant under the permutation between $\sigma(5)$ and $\sigma(6)$, so summation over this subgroup gives a factor of 2. Also if we restrict to a subgroup $S_{4}$ of $S_{6}$, which is the permutation between $\sigma(1),\sigma(2),\sigma(3),\sigma(4)$, then the summation over this subgroup gives a factor of 12. Since:
\begin{equation}
    6!=720=2\times 24\times 15=2\times 24\times C_{6}^{2}
\end{equation}
So this summation equals:
\begin{equation}
\begin{split}
    &2\times12\times\frac{c^{2}}{54}\mathop{\prod}\limits_{1\leq a<b\leq 6}z_{ab}^{2}\left(\mathop{\sum}\limits_{i<j}^{6}\frac{1}{z_{ij}^{2}}\mathop{\prod}\limits_{k\neq i,j}^{6}z_{ik}z_{jk}\right)\\
    &=\frac{4c^{2}}{9}\mathop{\prod}\limits_{1\leq a<b\leq 6}z_{ab}^{2}\left(\mathop{\sum}\limits_{i<j}^{6}\frac{1}{z_{ij}^{2}}\mathop{\prod}\limits_{k\neq i,j}^{6}z_{ik}z_{jk}\right)
\end{split}
\end{equation}
The summation result is the same for other two terms, so the final result is:
\begin{equation}
    \frac{4c^{2}}{3}\mathop{\prod}\limits_{1\leq a<b\leq 6}z_{ab}^{2}\left(\mathop{\sum}\limits_{i<j}^{6}\frac{1}{z_{ij}^{2}}\mathop{\prod}\limits_{k\neq i,j}^{6}z_{ik}z_{jk}\right)
\end{equation}
For term that proportional to $c$, the result is:
\begin{equation}
    \begin{aligned}
        &\frac{1}{6}\frac{c}{3}\mathop{\sum}\limits_{\sigma\in S_{6}}\mathop{\prod}\limits_{1\leq i<j\leq 6}z_{ij}^{2}(\frac{z_{\sigma(1)\sigma(3)}z_{\sigma(2)\sigma(4)}}{z_{\sigma(1)\sigma(2)}z_{\sigma(3)\sigma(4)}}\frac{z_{\sigma(1)\sigma(2)}z_{\sigma(5)\sigma(6)}}{z_{\sigma(1)\sigma(6)}z_{\sigma(2)\sigma(5)}}\frac{\mathop{\prod}\limits_{k\neq 5,6}z_{\sigma(k)\sigma(5)}z_{\sigma(k)\sigma(6)}}{z_{\sigma(5)\sigma(6)}^{2}}\\
        &+\frac{z_{\sigma(1)\sigma(5)}z_{\sigma(2)\sigma(6)}}{z_{\sigma(1)\sigma(2)}z_{\sigma(5)\sigma(6)}}\frac{z_{\sigma(3)\sigma(4)}z_{\sigma(5)\sigma(6)}}{z_{\sigma(3)\sigma(6)}z_{\sigma(4)\sigma(5)}}\frac{\mathop{\prod}\limits_{k\neq 3,4}z_{\sigma(k)\sigma(3)}z_{\sigma(k)\sigma(4)}}{z_{\sigma(3)\sigma(4)}^{2}}\\
        &+\frac{z_{\sigma(3)\sigma(5)}z_{\sigma(4)\sigma(6)}}{z_{\sigma(3)\sigma(4)}z_{\sigma(5)\sigma(6)}}\frac{z_{\sigma(1)\sigma(2)}z_{\sigma(3)\sigma(4)}}{z_{\sigma(1)\sigma(4)}z_{\sigma(2)\sigma(3)}}\frac{\mathop{\prod}\limits_{k\neq 1,2}z_{\sigma(k)\sigma(1)}z_{\sigma(k)\sigma(2)}}{z_{\sigma(1)\sigma(2)}^{2}})\\
    \end{aligned}
\end{equation}
Consider the term individually:
\begin{equation}
    \frac{1}{6}\frac{c}{3}\mathop{\sum}\limits_{\sigma\in S_{6}}\mathop{\prod}\limits_{1\leq i<j\leq 6}z_{ij}^{2}(\frac{z_{\sigma(1)\sigma(3)}z_{\sigma(2)\sigma(4)}}{z_{\sigma(1)\sigma(2)}z_{\sigma(3)\sigma(4)}}\frac{z_{\sigma(1)\sigma(2)}z_{\sigma(5)\sigma(6)}}{z_{\sigma(1)\sigma(6)}z_{\sigma(2)\sigma(5)}}\frac{\mathop{\prod}\limits_{k\neq 5,6}z_{\sigma(k)\sigma(5)}z_{\sigma(k)\sigma(6)}}{z_{\sigma(5)\sigma(6)}^{2}}
\end{equation}
For the permutation between $\sigma(5)$ and $\sigma(6)$, we have:
\begin{equation}
\begin{aligned}
     &\frac{z_{\sigma(1)\sigma(2)}z_{\sigma(5)\sigma(6)}}{z_{\sigma(1)\sigma(6)}z_{\sigma(2)\sigma(5)}}+\frac{z_{\sigma(1)\sigma(2)}z_{\sigma(6)\sigma(5)}}{z_{\sigma(1)\sigma(5)}z_{\sigma(2)\sigma(6)}}\\
     &=\frac{z_{\sigma(1)\sigma(2)}^{2}z_{\sigma(5)\sigma(6)}^{2}}{z_{\sigma(1)\sigma(5)}z_{\sigma(1)\sigma(6)}z_{\sigma(2)\sigma(5)}z_{\sigma(2)\sigma(6)}}
\end{aligned}
\end{equation}
For the permutation between $\sigma(3)$ and $\sigma(4)$, we have:
\begin{equation}
\begin{aligned}
     &\frac{z_{\sigma(1)\sigma(3)}z_{\sigma(2)\sigma(4)}}{z_{\sigma(1)\sigma(2)}z_{\sigma(3)\sigma(4)}}-\frac{z_{\sigma(1)\sigma(4)}z_{\sigma(2)\sigma(3)}}{z_{\sigma(1)\sigma(2)}z_{\sigma(3)\sigma(4)}}\\
     &=1
\end{aligned}
\end{equation}
For the permutation between $\sigma(1)$ and $\sigma(2)$, we have an extra factor of 2.
\begin{equation}
    720=2\times 2\times 2\times 90=2\times 2\times 2\times C_{6}^{2}\times C_{4}^{2}
\end{equation}
\begin{equation}
    \begin{aligned}
    &\frac{1}{6}\frac{c}{3}\mathop{\sum}\limits_{\sigma\in S_{6}}\mathop{\prod}\limits_{1\leq i<j\leq 6}z_{ij}^{2}(\frac{z_{\sigma(1)\sigma(3)}z_{\sigma(2)\sigma(4)}}{z_{\sigma(1)\sigma(2)}z_{\sigma(3)\sigma(4)}}\frac{z_{\sigma(1)\sigma(2)}z_{\sigma(5)\sigma(6)}}{z_{\sigma(1)\sigma(6)}z_{\sigma(2)\sigma(5)}}\frac{\mathop{\prod}\limits_{k\neq 5,6}z_{\sigma(k)\sigma(5)}z_{\sigma(k)\sigma(6)}}{z_{\sigma(5)\sigma(6)}^{2}}\\
        &=\frac{c}{9}\mathop{\prod}\limits_{1\leq i<j\leq 6}z_{ij}^{2}\left(\mathop{\sum}\limits_{a<b,m<n\neq a,b}^{6}z_{ab}^{2}\mathop{\prod}\limits_{k\neq a,b,m,n}^{6}z_{km}z_{kn}\right)
    \end{aligned}
\end{equation}
The result is the same for other two term, so the final result is:
\begin{equation}
\begin{aligned}
&\frac{1}{6}\frac{c}{3}\mathop{\sum}\limits_{\sigma\in S_{6}}\mathop{\prod}\limits_{1\leq i<j\leq 6}z_{ij}^{2}(\frac{z_{\sigma(1)\sigma(3)}z_{\sigma(2)\sigma(4)}}{z_{\sigma(1)\sigma(2)}z_{\sigma(3)\sigma(4)}}\frac{z_{\sigma(1)\sigma(2)}z_{\sigma(5)\sigma(6)}}{z_{\sigma(1)\sigma(6)}z_{\sigma(2)\sigma(5)}}\frac{\mathop{\prod}\limits_{k\neq 5,6}z_{\sigma(k)\sigma(5)}z_{\sigma(k)\sigma(6)}}{z_{\sigma(5)\sigma(6)}^{2}}\\
        &+\frac{z_{\sigma(1)\sigma(5)}z_{\sigma(2)\sigma(6)}}{z_{\sigma(1)\sigma(2)}z_{\sigma(5)\sigma(6)}}\frac{z_{\sigma(3)\sigma(4)}z_{\sigma(5)\sigma(6)}}{z_{\sigma(3)\sigma(6)}z_{\sigma(4)\sigma(5)}}\frac{\mathop{\prod}\limits_{k\neq 3,4}z_{\sigma(k)\sigma(3)}z_{\sigma(k)\sigma(4)}}{z_{\sigma(3)\sigma(4)}^{2}}\\
        &+\frac{z_{\sigma(3)\sigma(5)}z_{\sigma(4)\sigma(6)}}{z_{\sigma(3)\sigma(4)}z_{\sigma(5)\sigma(6)}}\frac{z_{\sigma(1)\sigma(2)}z_{\sigma(3)\sigma(4)}}{z_{\sigma(1)\sigma(4)}z_{\sigma(2)\sigma(3)}}\frac{\mathop{\prod}\limits_{k\neq 1,2}z_{\sigma(k)\sigma(1)}z_{\sigma(k)\sigma(2)}}{z_{\sigma(1)\sigma(2)}^{2}})\\
     &=\frac{c}{3}\mathop{\prod}\limits_{1\leq i<j\leq 6}z_{ij}^{2}\left(\mathop{\sum}\limits_{a<b,m<n\neq a,b}^{6}z_{ab}^{2}\mathop{\prod}\limits_{k\neq a,b,m,n}^{6}z_{km}z_{kn}\right)\\
     \end{aligned}
\end{equation}
which corresponds to the numerical results produced by Eq (\ref{Q1}).
For term that does not depend on $c$, one can check that it is actually antisymmetic and it will vanish over permutation sum.

So far we proved the equivalence for $n=2$ and $n=3$. For general $n$, we leave it as an exercise for enthusiastic readers to find an elementary proof of the equivalence between Simon's\cite{Simon_2009} and our formulae.

\bibliography{main.bib}
\bibliographystyle{JHEP}

\end{document}